\documentclass[letter]{article}

\usepackage{amsmath,amsfonts,amssymb,mathtools,stmaryrd}

\usepackage[letterpaper]{geometry}
\usepackage{ltexpprt}
\usepackage{xcolor}
\usepackage{hyperref}
\usepackage{cleveref}
\usepackage{enumitem} 

\usepackage[title]{appendix}
\usepackage{tikz}
\usepackage{soul}
\usepackage{authblk}

\usetikzlibrary{decorations.pathreplacing}
\usetikzlibrary{decorations.pathmorphing}

\newcommand*\samethanks[1][\value{footnote}]{\footnotemark[#1]}

\newcommand{\Ip}[2]{\langle #1,#2\rangle}

\newcommand{\rank}{\mathrm{rk}}
\newcommand{\nrank}{\rank^+}

\newcommand{\conv}{\mathrm{conv}}
\newcommand{\xc}{\mathrm{xc}}

\newcommand{\bE}{\mathbb{E}}
\newcommand{\bN}{\mathbb{N}}
\newcommand{\bR}{\mathbb{R}}

\newcommand{\cI}{\mathcal{I}}
\newcommand{\cL}{\mathcal{L}}
\newcommand{\cM}{\mathcal{M}}
\newcommand{\cP}{\mathcal{P}}
\newcommand{\cT}{\mathcal{T}}
\newcommand{\cR}{\mathcal{R}}

\newcommand{\bool}{\{0,1\}}

\newcommand{\boA}{\mathbf{A}}
\newcommand{\boB}{\mathbf{B}}
\newcommand{\boC}{\mathbf{C}}
\newcommand{\boD}{\mathbf{D}}
\newcommand{\boF}{\mathbf{F}}

\newcommand{\bone}{\mathbf{1}}

\newcommand{\Mall}{\cM_{\mathrm{all}}}
\newcommand{\Lall}{\cL_{\mathrm{all}}}
\newcommand{\eG}{\mathrm{GOOD}}
\newcommand{\eS}{\mathrm{SMALL}}
\newcommand{\eB}{\mathrm{BAD}}
\newcommand{\pos}{\mathrm{pos}}

\newcommand{\BalM}{\mathrm{Bal_M}}
\newcommand{\BalL}{\mathrm{Bal_L}}

\begin{document}

\title{The Exact Bipartite Matching Polytope Has Exponential Extension Complexity} 

\author{Xinrui Jia\thanks{Emails: \textit{xinrui.jia, ola.svensson, weiqiang.yuan@epfl.ch.} This work was supported by the Swiss National Science Foundation project 200021-184656 ``Randomness in Problem Instances and Randomized Algorithms".}, Ola Svensson\samethanks, Weiqiang Yuan\samethanks }
\affil{EPFL}
\date{}

\maketitle


\begin{abstract}
Given a graph with edges colored red or blue and an integer $k$, the 
\textit{exact perfect matching} problem asks if there exists a perfect matching with exactly $k$ red edges.
There exists a randomized polylogarithmic-time parallel algorithm to solve this problem, dating back to the eighties, but no deterministic polynomial-time algorithm is known, even for bipartite graphs.
In this paper we show that there is no sub-exponential sized linear program that can describe the convex hull of exact matchings in bipartite graphs. 
In fact, we prove something stronger, that there is no sub-exponential sized linear program to describe the convex hull of perfect matchings with an odd number of red edges.
\end{abstract}

\section{Introduction} \label{sec:intro}

The perfect matching problem is a central question in combinatorial optimization with a long and rich history. The first efficient (polynomial-time) algorithm for bipartite graphs can be traced back more than a century to Jacobi~\cite{Jacobi} and, for general graphs, a celebrated result by Edmonds~\cite{edm65matching} gave an efficient algorithm in 1965. Furthermore, a very recent breakthrough  shows that we can even solve the problem in bipartite graphs  in near-linear time~\cite{DBLP:journals/corr/abs-2203-00671}. 
We thus have a very good understanding of the basic problem. In contrast, and perhaps surprisingly, it remains a notorious problem to devise  efficient deterministic algorithms for slight variations. A prominent example is the \emph{exact perfect matching} problem: given an integer $k$ and a graph where each edge is colored either red or blue, the goal is to find a perfect matching with exactly $k$ red edges or to determine that such a matching does not exist.

This simple-to-state problem admits a very intriguing status that is related to the power of randomness in algorithms.
While there are several problems that admit efficient \emph{randomized} algorithms but for which no efficient deterministic algorithms are known, celebrated results in complexity theory strongly indicate that all efficient randomized algorithms can in fact be derandomized~\cite{ImpagliazzoW97}. 
Thus, our inability to find deterministic algorithms for problems that admit randomized ones  is very likely due to a lack of algorithmic techniques.
The gap in our understanding of the exact matching problem is especially large. If we restrict ourselves to deterministic algorithms, then efficient  algorithms are only known for very special graph families such as complete graphs, complete bipartite graphs~\cite{GKMT17}, bounded independence number graphs, and FPT algorithms~\cite{ElMaalouly},\cite{ElMaalouly-Steiner}. 
In particular, we have no deterministic \emph{sequential} algorithm that solves the above-mentioned problem  efficiently even on \emph{bipartite graphs}.
In contrast, if we allow for randomness, then we  have polylogarithmic-time \emph{parallel}  algorithms that  solve them on \emph{general graphs}.  Specifically,  
Mulmuley, Vazirani, and Vazirani~\cite{MVV87} gave a beautiful randomized efficient parallel algorithm for finding perfect matchings that generalizes nicely to variations such as the exact matching problem. 

The versatile randomized approach by Mulmuley, Vazirani, and Vazirani was published in 1987 and there was little progress in derandomizing it for decades. Then, in a relatively recent breakthrough~\cite{FennerGT16}, Fenner, Gurjar, Thierauf almost completely derandomized their approach for the perfect matching problem in bipartite graphs (almost since their deterministic parallel algorithm requires quasi-polynomially many processors instead of polynomially many). Building on the techniques in~\cite{FennerGT16}, the results have now been extended to perfect matchings in general graphs~\cite{nST17}, matroid intersection~\cite{GurjarT16}, and totally unimodular polytopes~\cite{GurjarTV18}. In all of these works, our excellent polyhedral understanding of these problems have been crucial in the derandomization of~\cite{MVV87}. A major difficulty in extending these techniques to obtain a deterministic algorithm for the exact matching problem is that we do not have a good polyhedral understanding, i.e.,  an amenable description of the convex hull of the indicator vectors of exact perfect matchings. The goal of this paper is  to advance our polyhedral understanding of the exact matching problem. 

The convex hull of perfect matchings in bipartite graphs has a very simple and small (linearly many inequalities) description.  But what happens to the exact matching problem in bipartite graphs? In other words, how much more complex is this polytope when we introduce the constraint that the perfect matchings must have exactly $k$ red edges?  
%
%
In this paper, we show that the description of the polytope becomes significantly more complex and the exact perfect matching problem cannot be solved by a linear program of subexponential size even for bipartite graphs\footnote{We remark that this question only makes sense for bipartite graphs and not for general graphs due to the breakthrough result by Rothvoss~\cite{Rothvoss17}: already the perfect matching problem has exponential extension complexity in general graphs.}. More formally, we are given an instance $\cI_{n,k}=(G_n=(U,V,E),k)$, where  $|U|=|V|=n$, and there are two parallel edges, one each of red and blue, between every pair $u\in U$ and $v\in V$. 
We define the associated polytope of $\cI_{n,k}$ as the convex hull of the characteristic vectors of all perfect matchings with exactly $k$ red edges, i.e. $P_{\cI_{n,k}}:=\conv(\{\chi^M:M \text{ is a perfect matching}\land |\{e:e\in M\land e\text{ is red}\,\}|=k \})\subseteq \bR^E$. Our main result is that the polytope $P_{\cI_{n,k}}$ has exponential extension complexity.

\begin{theorem}\label{thm:xc_exact}
For all $n>0$, there exists some $k:=k(n)$ such that $\xc(P_{\cI_{n,k}})=2^{\Omega(n)}$.
\end{theorem}
\begin{Remark}
 Theorem~\ref{thm:xc_exact} can be strengthened in the following way: fix any constant $\epsilon>0$, then for all $n>0$ and $\epsilon n\le k\le (1-\epsilon)n$, $\xc(P_{\cI_{n,k}})=2^{\Omega(n)}$. This stronger result directly follows from the simple observation that for any $n,k,n',k'$ such that $k'\ge k,n'-k'\ge n-k$, $P_{\cI_{n,k}}$ is isomorphic to some face of $P_{\cI_{n',k'}}$.
\end{Remark}

In fact, our paper proves a stronger result. Consider the following problem: we are given a bipartite graph where each edge is colored either red or blue. Instead of finding a perfect matching with exactly $k$ red edges, we aim to find a perfect matching with an odd number of red edges. We call this problem \emph{parity bipartite perfect matching}. Given the graph $G_n=(U,V,E)$, let $\Mall(G_n)$ denote the set of all perfect matchings of $G_n$ with an odd number of red edges. The (odd-)parity bipartite perfect matching polytope associated with $G_n$ is defined as $P_{G_n}:=\conv(\{\chi^M:M\in \Mall(G_n)\})$. We show that this polytope has exponential extension complexity.

\begin{theorem}\label{thm:xc_odd}
For all $n>0$, $\xc(P_{G_{n}})=2^{\Omega(n)}$.
\end{theorem}

Note that $P_{G_n}=\conv(\cup_{k \text{ odd}} P_{\cI_{n,k}})$. Using the following basic fact in~\cite{Balas1998} on the extension complexity of the convex hull of a union of polytopes, one can directly deduce Theorem~\ref{thm:xc_exact} from Theorem~\ref{thm:xc_odd}.
\begin{fact}[\cite{Balas1998}]
Let $P_1,\ldots,P_k$ be a set of polytopes. Then $\xc(\cup_{1\le i\le k} P_i)\le \sum_{i=1}^{k} \xc(P_i)$.
\end{fact}

\begin{Remark}
First, our result also applies to perfect matching with an even (instead of odd) number of red edges (See Claim~\ref{clm:parity} and  Remark~\ref{rmk:even-matching}). Second, our result also applies to simple graphs because we can replace parallel edges with a gadget to get a simple graph.
The gadget is as follows: replace each parallel edge $(u,v)$ with a path of length three by adding two vertices, $x_r$, $y_r$ for the vertices subdividing the red edge and $x_b$, $y_b$ for subdividing the blue edge.
The color of the new edge $(u, x_r)$ is red and all others are blue. 
Then if neither color edge of $(u,v)$ is used in a perfect matching in the multigraph it means that vertices $u$ and $v$ are both matched by other edges so we can find a perfect matching in the simple graph by taking $(x_r,y_r)$ and $(x_b,y_b)$ of the simple graph, and this doesn't change the number of red edges.
If the blue edge of $(u,v)$ is taken in the multigraph then in the simple graph we take $(u, x_b)$, $(y_b, v)$, and $(x_r, y_r)$.
Otherwise take $(u, x_r)$, $(y_r, v)$, and $(x_b, y_b)$, and note that we take one red edge in the simple graph exactly when one red edge is taken in the multigraph.
It is not difficult to see that this gadget reduction preserves the extension complexity.
In the simple graph we can let $x_{(u,v)}$ be the indicator variable for $(u, x_r)$ (and $(y_r, v)$) and then $1 - x_{(u,v)}$ is the indicator variable for $(x_r, y_r)$ since $(u, x_r)$ and $(y_r, v)$ are always taken together, exactly when $(x_r, y_r)$ is not taken.
Then it is clear that the polytope for perfect matchings with an odd number of red edges in the simple graph can be projected down to the original multigraph and thus has no smaller extension complexity (in fact they are isomorphic so they have the same extension complexity).
\end{Remark}

\paragraph{Further related work.}
While we show that the extension complexity of the parity bipartite perfect matching problem is exponential, the weighted version of parity bipartite perfect matching admits a deterministic polynomial-time algorithm. 
One first computes a min-weight perfect matching, and should this matching have an even number of red edges, then compute a min-weight alternating cycle with an odd number of red edges and output the symmetric difference. This can be done in polynomial time which is implied by a more general result by Artmann, Weismantel, and Zenklusen~\cite{AWZ17} on so-called network matrices, which also capture bipartite matchings.
\cite{AWZ17} also gave a strongly polynomial-time algorithm for solving integer programs where the constraint matrix with rank $n$ is \textit{bimodular}, i.e. all determinants of $n \times n$ sub-matrices are bounded by 2 in absolute value.
Their first step is to reduce bimodular integer programming to a totally unimodular integer linear program with an additional parity constraint, which is a generalization of the parity matching problem in bipartite graphs.
While it is known that certain such problems have exponential extension complexity~\cite{Cevallos_2018} via a close connection to the perfect matching problem in general graphs, it is unlikely that there is a simple reduction to our problem. In particular,   
the algorithm for totally unimodular matrices with a parity constraint is much more complex than the one for the parity bipartite perfect matching problem.

Other work related to the extension complexity of the matching polytope include~\cite{BraunPokutta} in which the authors generalize the result of~\cite{Rothvoss17} by exploring the approximability of the matching polytope by a linear program. A tight bound is obtained for a $(1+\epsilon)$-approximation when $\epsilon \in O(1/n)$. The work in~\cite{Sinha} extends the tight lower bounds to all $\epsilon$, $2/n \leq \epsilon \leq 1$.

\subsection{Proof Overview}

Our proof of the exponential extension complexity of the parity bipartite perfect matching polytope follows the recipe of Rothvoss's proof of the exponential extension complexity of the matching polytope~\cite{Rothvoss17}.
The end goal is to demonstrate an exponential lower bound on the non-negative rank of a slack matrix of inequalities valid for the polytope.
However, we have three fundamental differences when applying this recipe to our problem.
The first is that while the polytope of matchings in general graphs is well-understood (by the work of Edmonds~\cite{edm65}) and automatically provides a set of inequalities to be used in the slack matrix, there was no such set of inequalities known for the parity bipartite matching problem until this work.
In addition to the standard vertex-degree constraints of the bipartite perfect matching polytope, we have inequalities resulting from what we call \textit{labels} and these are used to exhibit an exponential lower bound on the non-negative rank of the slack matrix.
In Section~\ref{sec:relax-poly}, we describe this relaxation in detail and in {Section}~\ref{sec:append} we prove that our relaxation is empty if and only if the graph has no perfect matching with an odd number of red edges.
We leave it as an interesting open problem to resolve whether or not our relaxation is in fact integral.



The second issue we now face is that we have a problem in a \textit{bipartite} graph.
The convex hull of integral matchings in bipartite graphs can be expressed by just linearly many inequalities and we lose the structure of odd-cuts that Rothvoss uses in his slack matrix. 
In other words, we have to define our own concept of partitions that exploits the characteristics of the parity problem.
Our partitions are defined in Section~\ref{subsec:partitions} and the remainder of Section~\ref{sec:main} proves the main technical ingredient, Lemma~\ref{lem:over_cover}, which states that large-enough rectangles cover too many entries of large slack to be able to decompose into a subexponential-sized non-negative factorization.
In order to follow the framework of Rothvoss, we need our definition of partitions and \textit{triples} of sets to allow us to generate entries in the slack matrix.
For a fixed triple and fixed rectangle, we can look at the probability that a uniformly random matching (or labeling) selected from all matchings (or labelings) that \textit{respect} the triple lies in the rectangle.
In particular, these triples need to be defined in such a way that the measure of small- and large-slack entries is an expectation over all partitions of the \textit{product} of the above described probabilities. 
Then instead of directly comparing the entries in the slack matrix we can compare the contribution of triples to small- and large-slack entries.
We use the same naming convention as Rothvoss and split our triples into three categories: good, small, and bad.
Roughly speaking, good triples are those where generated matchings and labelings are almost uniform, small triples have small contribution anyway, and bad triples are neither of the first two.
Their formal treatments can be found in Sections~\ref{subsec:classes},~\ref{subsec:good}, and~\ref{subsec:bad}, respectively.

Our third challenge is a technical one and is regarding the contribution of good triples.
The details can be found in Section~\ref{subsec:good}.
Rothvoss uses a partition structure given in a pair of sets and is able to show that for a large enough non-zero rectangle, the second set in good pairs that generate slack-2 entries cannot overlap on more than one element.
Hence, large-enough rectangles cover too few small-slack entries to be decomposed into a small non-negative factorization of the slack matrix.
This argument does not quite work for us.
We require an extra condition on the structure of the triples to be able to use this approach, which we get in part as a consequence of the partitions we set up.
We also need to take permutations on the part of the partition that generates positive slack entries and use the probabilistic method to show that a certain family of permutations exists.



\section{Preliminaries} \label{sec:prel}
In our paper, a set can be either ordered or unordered. We will use bold letters to denote sequences (ordered sets) and non-bold letters to denote unordered sets. For example, $\mathbf{A}$ is a sequence while $B$ is an unordered set.

We specify each edge $e\in E$ by a triple $(u,v,c)$, where $u\in U$ and $v\in V$ are two vertices and $c$ is the color of the edge (either red or blue).

The extension complexity of a polytope $P$, denoted by $\xc(P)$, is the smallest integer $r>0$ such that $P$ is a linear projection of some $r$-facet polytope $Q$.

Let $P$ be a polytope with vertices $x_1,\ldots,x_l$, i.e. $P=\conv(x_1,\ldots,x_l)$, and let $C$ be a set of inequalities $C=\{a_i^{\top}x\le b_i:1\le i\le m\}$ such that $P$ is {feasible w.r.t.} all the inequalities in $C$, i.e. $P\subseteq\{x:a_i^{\top}x\le b_i,\forall 1\le i\le m\}$. Define the slack matrix of $P$ w.r.t. $C$, denoted $S(P,C)$, as an $l\times m$ non-negative real matrix with $S(P,C)_{i,j}:=b_j-a_j^{\top}x_i$.

Given a non-negative rank matrix $S\in \bR^{n\times m}_{\ge 0}$, we use $\nrank(S)$ to denote the non-negative rank of $S$, which is defined as the smallest positive integer $r$ such that $S$ can be written as the sum of $r$ non-negative rank-1 matrices, or equivalently, $S=\sum_{i\in [r]} u_iv^{\top}_i$, where $u_i\in \bR^{n}_{\ge 0},v_{i}\in \bR^m_{\ge 0}$ are non-negative vectors.

\begin{theorem}[\cite{Yannak88}]
Given a polytope $P$ and a set of inequalities $C=\{a_i^{\top}x\le b_i:1\le i\le m\}$, if the inequalities in $C$ exactly characterize $P$, i.e. $P=\{x:a_i^{\top}x\le b_i,\forall 1\le i\le m\}$, then $\xc(P)=\nrank(S(P,C))$.
\end{theorem}

Note that the above theorem holds for any linear inequality characterization of $P$. This implies that the non-negative rank of the slack matrix associated with a loose linear inequality characterization is a lower bound for the extension complexity, since we can make the characterization tight by adding inequalities and the non-negative rank of a matrix is no smaller than that of its submatrix.

\begin{corollary} \label{coro:nrlbxc}
Given a polytope $P$ and a set of inequalities $C=\{a_i^{\top}x\le b_i:1\le i\le m\}$, if $P$ is {feasible w.r.t.} all the inequalities in $C$, then $\nrank(S(P,C))\le \xc(P)$.
\end{corollary}

Exactly characterizing the non-negative rank of a matrix is usually difficult. In~\cite{FMP+12}, the authors use the covering number to give a lower bound for the non-negative rank. Unfortunately, in our setting, the covering number is too loose.
We will use the following stronger tool from~\cite{Rothvoss17} based on hyperplane separation.

\begin{lemma}[\cite{Rothvoss17}]\label{lem:hyperplane}
Let $S\in \bR_{\ge 0}^{l\times m}$ be a non-negative real matrix. Then for any real matrix $W\in \bR^{l\times m}$ with at least one positive entry,
\begin{equation*}
\nrank(S)\ge \frac{\Ip{W}{S}}{\|S\|_{\infty}\cdot \max\{\Ip{W}{uv^{\top}}:u\in \{0,1\}^l,v\in \{0,1\}^m\}}.
\end{equation*}
Here, $\Ip{W}{S} := \sum_{i=1}^l \sum_{j=1}^m S_{ij} \cdot W_{ij}$ is the Frobenius inner product.
\end{lemma}

\section{Relaxation of the Polytope}\label{sec:relax-poly}

In this section, we exhibit a set of inequalities with a nice structure that is valid for the parity bipartite perfect matching polytope and suffices to show this polytope has exponential extension complexity.
The intuition behind our relaxation is the following basic idea: if we have a collection of disjoint 4-cycles where each cycle has alternating red and blue edges then it is clear that there is no perfect matching with an odd number of red edges.
Similarly, take the bipartite graph in Figure~\ref{fig:intuit} where we partition $U$ into $U_0 \sqcup U_1$ and $V$ into $V_0 \sqcup V_1$ where for simplicity suppose for now that $|U_0| = |V_1|$.
Any perfect matching must take the same number of edges from $(U_0, V_0)$ as from $(U_1, V_1)$ since $|U_0| = |V_1|$. 
Hence, if all edges in $(U_0, V_0)$ and $(U_1, V_1)$ are red and all others are blue, this graph does not have a perfect matching with an odd number of red edges.
We now formally give the relaxation in full generality which is based on this construction.

As aforementioned, our inequalities come from the vertex labelings of $G$. Define
\begin{equation*}
\Lall(G):=\{L:U\cup V\rightarrow \{0,1\} : |L^{-1}(1)|\equiv n\pmod 2\}
\end{equation*}
as the set of 0-1 labelings for which the number of ones has the same parity as $n$. We associate each labeling $L\in \Lall(G)$ with a set of edges
\begin{equation*}
E_{L}:=\{e=(u,v,c):L(u)=L(v),c= \mbox{blue}\}\cup \{e=(u,v,c):L(u)\ne L(v),c= \mbox{red}\}.
\end{equation*}
We claim that for any perfect matching $M$ with an odd number of red edges, there must be some edge in $M$ that belongs to $E_{L}$ for every labeling $L$.
We say a colored edge $e=(u,v,c)$ \textit{violates} a labeling $L$ if $e \in E_L$.
Otherwise, we say that it is \textit{correct} or \textit{consistent} with respect to a labelling.

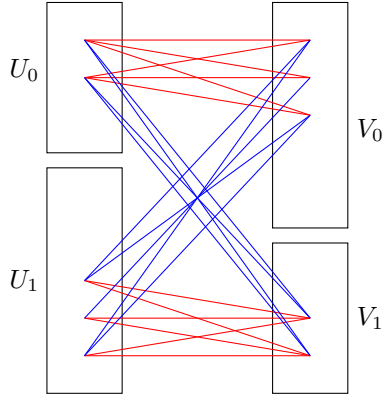
\begin{figure}
    \centering
    \begin{tikzpicture}
    \draw[draw=black] (0, 0) rectangle ++(1,3);
    \draw[draw=black] (0, 3.2) rectangle ++(1,2);
    \draw[draw=black] (3, 0) rectangle ++(1,2);
    \draw[draw=black] (3, 2.2) rectangle ++(1,3);
    
    \draw[color=red] (0.5, 0.5) -- (3.5, 1);
    \draw[color=red] (0.5, 0.5) -- (3.5, 0.5);
    \draw[color=red] (0.5, 1) -- (3.5, 1);
    \draw[color=red] (0.5, 1) -- (3.5, 0.5);
    \draw[color=red] (0.5, 1.5) -- (3.5, 0.5);
    \draw[color=red] (0.5, 1.5) -- (3.5, 1);
    
    \draw[color=red] (0.5, 4.2) -- (3.5, 3.7);
    \draw[color=red] (0.5, 4.2) -- (3.5, 4.2);
    \draw[color=red] (0.5, 4.2) -- (3.5, 4.7);
    \draw[color=red] (0.5, 4.7) -- (3.5, 3.7);
    \draw[color=red] (0.5, 4.7) -- (3.5, 4.2);
    \draw[color=red] (0.5, 4.7) -- (3.5, 4.7);
    
    \draw[color=blue] (0.5, 0.5) -- (3.5, 3.7);
    \draw[color=blue] (0.5, 1) -- (3.5, 4.2);
    \draw[color=blue] (0.5, 1.5) -- (3.5, 4.7);
    \draw[color=blue] (0.5, 4.2) -- (3.5, 0.5);
    \draw[color=blue] (0.5, 4.7) -- (3.5, 1);
    
    \draw[color=blue] (0.5, 0.5) -- (3.5, 4.7);
    \draw[color=blue] (0.5, 1.5) -- (3.5, 3.7);
    \draw[color=blue] (0.5, 4.7) -- (3.5, 0.5);
    \draw[color=blue] (0.5, 4.2) -- (3.5, 1);
    
    \node at (-0.3, 4.3) {$U_0$};
    \node at (4.3, 3.5) {$V_0$};
    \node at (-0.3, 1.5) {$U_1$};
    \node at (4.3, 1) {$V_1$};

    \end{tikzpicture}
    \caption{This graph does not have a perfect matching with an odd number of red edges.}
    \label{fig:intuit}
\end{figure}

\begin{claim}\label{clm:parity}
Let $M\in \Mall(G)$ and $L\in \Lall(G)$. Then $|M\cap E_L|\ge 1$.
\end{claim}
\begin{proof}
For a contradiction, suppose $M\subseteq E\setminus E_L$. Let $U_t:=\{u\in U:L(u)=t\}$ and $V_t:=\{v\in V:L(v)=t\}$. 
See Figure~\ref{fig:intuit} for a diagram.
Let $x:=|M\cap E(U_1,V_1)|$ be the number of edges between $U_1$ and $V_1$ in $M$. Since $M$ is perfect, $|M\cap E(U_0,V_1)|=|M\cap E(U,V_1)|-|M\cap E(U_1,V_1)|=|V_1|-x$. By the same argument, we have $|M\cap E(U_0,V_0)|=(n-|U_1|)-(|V_1|-x)$.

Since $M\subseteq E\setminus E_L$, $\{e\in M:e\text{ is red}\}=(M\cap E(U_1,V_1)) \cup (M\cap E(U_0,V_0))$. Thus
\begin{equation}\label{eq:parity}
   |\{e\in M:e\text{ is red}\}|=n-|U_1|-|V_1|+2x
\end{equation}
is even since $n$ and $|U_1|+|V_1|=|L^{-1}(1)|$ have the same parity, leading to a contradiction.
\end{proof}

By Claim~\ref{clm:parity}, the following relaxation is valid for the parity bipartite perfect matching polytope.
\begin{align*}
    \sum_{e\in \delta(u)} x_{e} &= 1 \quad \mbox{for all } u \in U\cup V \\
    \sum_{e \in E_L} x_{e} &\geq 1 \quad \mbox{for all } L \in \Lall(G) \tag{$C_G$}\\
    x_{e} &\geq 0 \quad \mbox{for all $e \in E$}.
\end{align*}

\begin{Remark}\label{rmk:even-matching}
Note that the relaxation can be formulated for an even number of red edges by letting $\Lall(G)$ to be the set of labels such that $|L^{-1}(1)|$ and $n$ have different parities for all $L \in \Lall(G)$ .
\end{Remark}

Let $C_{G}$ be the set of constraints above.
Recall that $P_G$ is the parity bipartite perfect matching polytope associated with $G$. As Theorem~\ref{clm:parity} shows, $P_{G}$ is {feasible w.r.t.} all the inequalities in $C_{G}$.
By~\Cref{coro:nrlbxc}, to prove Theorem~\ref{thm:xc_odd}, it suffices to show the following theorem.
\begin{theorem}\label{thm:nr_relax}
For all $n>0$, $\nrank(S(P_{G},C_{G}))=2^{\Omega(n)}$.
\end{theorem}

Now we can see why the covering number is loose.
In fact, it is very similar to the observation made in~\cite{Yannak88} that the covering number is insufficient to show that the perfect matching polytope has exponential extention complexity.
Consider only the labeling constraints. 
For any pair $e_1$, $e_2$ of non-adjacent edges, let $\mathcal{L}_{e_1, e_2}$ be the set of labelings that are violated by $e_1$ and $e_2$ and let $\mathcal{M}_{e_1, e_2}$ be the set of perfect matchings that contain $e_1$ and $e_2$. 
By Claim~\ref{clm:parity} we have that the slack entry $(M,L)$ is non-zero for any $M \in \mathcal{M}_{e_1, e_2}$ and $L \in \mathcal{L}_{e_1, e_2}$ so the $O(n^4)$ rectangles of the form $\mathcal{L}_{e_1, e_2} \times \mathcal{M}_{e_1, e_2}$ are a valid rectangle covering.
However, an entry with slack $\ell$ lies in ${\ell+1 \choose 2}$ rectangles since there are $\ell + 1$ edges that are violating.
So we do not get a valid non-negative factorization of the slack matrix as $O(n^4)$ many $0/1$ rank-1 matrices, for the same reason as in~\cite{Rothvoss17} that large entries are over-covered.
Nonetheless, we will use the hyperplane separation bound to give an exponential lower bound on the extension complexity of the parity bipartite matching polytope.

\section{Proof of Theorem~\ref{thm:nr_relax}}\label{sec:main}

In this section we prove Theorem~\ref{thm:nr_relax}. From now on, we will focus on the case when $n=4k(2m+1)+3$ is odd, where $k\ge 400$ is some large constant and $n$ grows with $m$. In particular, our proof works only when $m$ is sufficiently large, say greater than $20\cdot(4k)!$.

As mentioned in the previous section, our proof for Theorem~\ref{thm:nr_relax} relies on~\Cref{lem:hyperplane}. 
In~\Cref{subsec:weight_func}, we construct the weight function $W$ which gives an exponential lower bound in~\Cref{lem:hyperplane}. 
The weight function consists of two key components: $\mu_3$ and $\mu_{4k+3}$ which are distributions on slack-$2$ and slack-$(4k+2)$ entries respectively. 
Then we convert our problem into finding $\mu_3,\mu_{4k+3}$ and proving that any large rectangle w.r.t. $\mu_3$ disjoint from slack-$0$ entries must over-cover $\mu_{4k+3}$.
Before describing $\mu_3$ and $\mu_{4k+3}$, we introduce the notion of partitions in Section~\ref{subsec:partitions}, which will be crucial in generating the distributions.
The idea is that for any fixed partition, we can compare the measures $\mu_3$ and $\mu_{4k+3}$ on the entries of a rectangle that are induced by the partition.
In Section~\ref{subsec:generate}, we use partitions to express $\mu_3$ and $\mu_{4k+3}$ as sums of product distributions.
In particular, each product distribution for $\mu_3$ is represented by a triple $(T,H,\boD)$. 
To give an upper bound on $\mu_3$ in terms of $\mu_{4k+3}$, in~\Cref{subsec:classes}, we classify the triples into three categories: good, small, and bad, and then give a one line bound for the contribution of small triples. 
We will then bound the contribution of good triples and bad triples in~\Cref{subsec:good} and~\ref{subsec:bad}, respectively.

\subsection{The Weight Function for~\Cref{lem:hyperplane}}\label{subsec:weight_func}
We will use \Cref{lem:hyperplane} to show Theorem~\ref{thm:nr_relax}. First, let us construct the weight function $W$ in the hyperplane separation lemma.

For any $M\in \Mall(G)$ and $L\in \Lall(G)$, we call $(M,L)$ an $l$-violation pair if $|M\cap E_L|=l$, i.e., it is a slack-$(l-1)$ entry. Our ``hard'' weight function $W$ is defined as follows:
\begin{equation*}
    W_{M,L}:=\left\{\begin{matrix}
    -\infty & |M\cap E_L|=1 \\
    \mu_3(M,L) & |M\cap E_L|=3\\
    -\frac{1}{4k+2}\cdot\mu_{4k+3}(M,L) & |M\cap E_L|=4k+3\\
    0 & \text{otherwise},
    \end{matrix}\right.
\end{equation*}
where $\mu_3$ and $\mu_{4k+3}$ are some distributions on $3$-violation and $(4k+3)$-violation pairs respectively, which will be specified later.

It is easy to verify that
\begin{equation}\label{equ:wt_func}
  \Ip{W}{S(P_G, C_G)}=0\cdot{-\infty}+(3-1)-\frac{(4k+3)-1}{4k+2}=1. 
\end{equation}

We say that a two-party function $\cR:\Mall(G)\times \Lall(G)\rightarrow \bool$ is a rectangle if $\cR=uv^{\top}$ for some boolean vectors $u\in \bool^{\Mall(G)}$ and $v\in \bool^{\Lall(G)}$. Given a distribution $\mu$ over $\Mall(G)\times \Lall(G)$, we define $\mu(\cR):=\Ip{\mu}{\cR}$ as the probability that a pair $(M,L)$ sampled from $\mu$ belongs to $\cR$. We show that the inner product of $W$ and any rectangle is tiny.
\begin{lemma}\label{lem:rec}
Let $\cR:\Mall(G)\times \Lall(G)\rightarrow \bool$ be a rectangle. Then $\Ip{W}{\cR}\le 2^{-\gamma m+1}$ for some constant $\gamma:=\gamma(k)>0$.
\end{lemma}

\Cref{equ:wt_func} and~\Cref{lem:rec} directly imply Theorem~\ref{thm:nr_relax}:
\begin{equation*}
    \nrank(S(P_G,C_G)) \ge  \frac{\overbrace{\Ip{W}{S(P_G,C_G)} }^{=1}}{\underbrace{\|S(P_G,C_G)\|_{\infty}}_{\le n-1}\cdot \underbrace{\max\{\Ip{W}{\cR}:\cR\text{ is a rectangle}\}}_{\le 2^{-\gamma m+1}}}=2^{\Omega(n)}.
\end{equation*}

To show \Cref{lem:rec}, we only need to show:
\begin{lemma}\label{lem:over_cover}
Let $\cR:\Mall(G)\times \Lall(G)\rightarrow \bool$ be a rectangle. If for all $(M,L)\in \cR$, $|M\cap E_L|>1$, then $\mu_3(\cR)\le \frac{40}{k^2}\cdot \mu_{4k+3}(\cR)+2^{-\gamma m}$ when $m$ is sufficiently large.
\end{lemma}
\begin{proof}[Proof of \Cref{lem:rec} assuming \Cref{lem:over_cover}]
If $\cR$ contains a $1$-violation entry, by definition, $\Ip{W}{\cR}=-\infty$. Otherwise, when $k\ge 400$, we have
\begin{equation*}
    \Ip{W}{\cR}=2\mu_3(\cR)-\frac{1}{4k+2}\mu_{4k+3}(\cR)\le \left(\frac{80}{k^2}-\frac{1}{4k+2}\right)\mu_{4k+3}(\cR)+2^{-\gamma m+1}\le 2^{-\gamma m+1}.
\end{equation*}
\end{proof}
Intuitively, the above lemma is saying that any ``large'' rectangle disjoint from $1$-violation entries must over-cover $(4k+3)$-violation entries, where by ``large'' rectangles we mean rectangles that cover at least a $2^{-\gamma m+2}$-fraction of $3$-violation entries according to $\mu_3$.

In the rest of this section, we prove \Cref{lem:over_cover}.
\subsection{Partitions}\label{subsec:partitions}

\begin{figure}
    \centering
    \begin{tikzpicture}
    \draw[draw=black, fill=red!10] (0, 0) rectangle ++(5, 3.5);
    \draw[draw=black, fill=green!10] (6, -1) rectangle ++(1, 4.5);
    \draw[draw=black, fill=blue!10] (8, 0) rectangle ++(5, 3.5);
    
    \draw[draw=black, fill=red!20] (0.25, 0.25) rectangle ++(1, 3);
    \draw[draw=black, fill=red!20] (1.5, 0.25) rectangle ++(1, 3);
    \draw [fill=black] (2.9, 1.75) circle [radius = 0.02];
    \draw [fill=black] (3.1, 1.75) circle [radius = 0.02];
    \draw [fill=black] (3.3, 1.75) circle [radius = 0.02];
    \draw[draw=black, fill=red!20] (3.75, 0.25) rectangle ++(1, 3);
    
    \draw[draw=black, fill=blue!20] (8.25, 0.25) rectangle ++(1, 3);
    \draw[draw=black, fill=blue!20] (9.5, 0.25) rectangle ++(1, 3);
    \draw [fill=black] (10.9, 1.75) circle [radius = 0.02];
    \draw [fill=black] (11.1, 1.75) circle [radius = 0.02];
    \draw [fill=black] (11.3, 1.75) circle [radius = 0.02];
    \draw[draw=black, fill=blue!20] (11.75, 0.25) rectangle ++(1, 3);
    
    \foreach \x in {0,...,7}
        {
        \draw [fill=black] (0.5, 0.5 + 0.35*\x) circle [radius = 0.05];
        \draw [fill=black] (1, 0.5 + 0.35*\x) circle [radius = 0.05];
        
        \draw [fill=black] (1.75, 0.5 + 0.35*\x) circle [radius = 0.05];
        \draw [fill=black] (2.25, 0.5 + 0.35*\x) circle [radius = 0.05];
        
        \draw [fill=black] (4, 0.5 + 0.35*\x) circle [radius = 0.05];
        \draw [fill=black] (4.5, 0.5 + 0.35*\x) circle [radius = 0.05];
        }

    \foreach \x in {0,...,7}
        {
        \draw [fill=black] (8.5, 0.5 + 0.35*\x) circle [radius = 0.05];
        \draw [fill=black] (9, 0.5 + 0.35*\x) circle [radius = 0.05];
        
        \draw [fill=black] (9.75, 0.5 + 0.35*\x) circle [radius = 0.05];
        \draw [fill=black] (10.25, 0.5 + 0.35*\x) circle [radius = 0.05];
        
        \draw [fill=black] (12, 0.5 + 0.35*\x) circle [radius = 0.05];
        \draw [fill=black] (12.5, 0.5 + 0.35*\x) circle [radius = 0.05];
        }
        
    \foreach \x in {0,...,7}
        {
        \draw [fill=black] (6.25, -0.75 + 0.35*\x) circle [radius = 0.05];
        \draw [fill=black] (6.75, -0.75 + 0.35*\x) circle [radius = 0.05];

        }
    \foreach \x in {0,...,2}
        {
        \draw [fill=black] (6.25, 2.5 + 0.35*\x) circle [radius = 0.05];
        \draw [fill=black] (6.75, 2.5 + 0.35*\x) circle [radius = 0.05];
        }
        
    \node at (2.5, 4) {$A$};
    \node at (10.5, 4) {$B$};
    \node at (6.5, 4) {$C$};
    
    \end{tikzpicture}
    \caption{Visualization of a partition $T = (A, C, B)$. Each of the $m$ blocks in $A$ and $B$ have $4k$ pairs of vertices while $C$ has $4k+3$ pairs.}
    \label{fig:partition}
\end{figure}
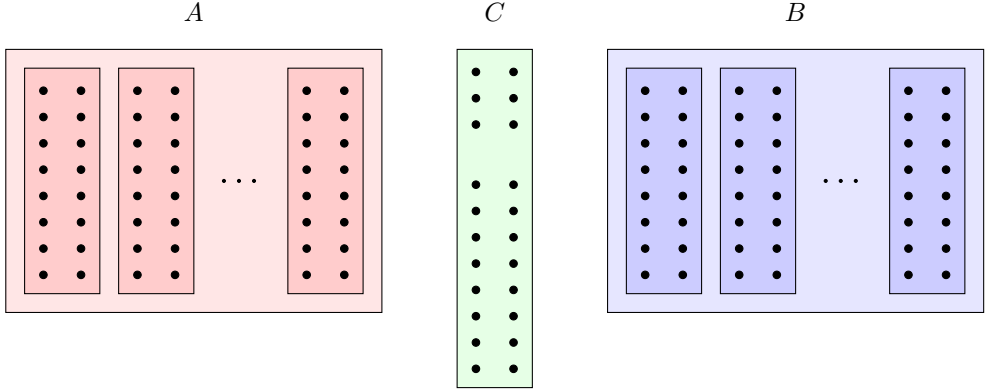

The vertices of the graph are labeled $u_1,\ldots,u_n\in U$ and $v_1,\ldots,v_n\in V$. For each $i\in [n]$, $(u_i,v_i)$ are naturally made into pairs. Note that a vertex might not be matched with the vertex in the same pair when talking about $\mu_3$ and $\mu_{4k+3}$ later.

Similar to Rothvoss's partition, our partition consists of $A=\{\boA_1,\ldots,\boA_m\}$, which will be used to generate almost uniform labelings, and $B=\{\boB_1,\ldots,\boB_m\}$, which will be used to generate almost uniform matchings. We also have a set $C$ (in Rothvoss's proof $C$ and $D$ are used), the only part that causes violations.

Unlike Rothvoss's partition, since our graph is colored, each block $\boA_i$ or $\boB_i$ cannot contain all the labelings or all the matchings. Thus, in each block, vertices are not ``identical''. In fact, vertices (or more precisely, pairs) are ``ordered'' in each $\boA_i$ and $\boB_i$ in our partition ($C$ is not ordered at this stage).

Formally, a partition $T$ is defined as $T:=(A=\{\boA_1,\ldots,\boA_m\},C,B=\{\boB_1,\ldots,\boB_m\})$ where
\begin{itemize}
    \item $A$ contains $m$ disjoint blocks $\boA_1,\ldots,\boA_m$ where each $\boA_i$ is a sequence of $4k$ pairs of vertices.
    \item $B$ contains $m$ disjoint blocks $\boB_1,\ldots,\boB_m$ where each $\boB_i$ is a sequence of $4k$ pairs of vertices.
    \item $C$ is an unordered set of the remaining $4k+3$ pairs of vertices.
\end{itemize}

Given an ordered set of pairs of vertices $\boA$ (resp. an unordered set $A$), let $G_{\boA}$ (resp. $G_A$), be a subgraph induced by all the vertices in $\boA$ (resp. $A$). Given an $M\in\Mall(G)$ and $L\in \Lall(G)$, let $M_{\boA}$ (resp. $M_A$) 
be the matching resulting from restricting $M$ on $G_{\boA}$ (resp. $G_A$). Let $L_{\boA}$ (resp. $L_A$) be the labeling resulting from restricting $L$ on $G_{\boA}$ (resp. $G_A$). For $M$ a matching that is perfect on $G_A$, we say $M_A$ is consistent with $L_A$ if all the edges in $M_A$ are correct w.r.t. $L_A$.

Now let us see what is inside $\boA_i$ and $\boB_j$. Given a sequence of $4k$ pairs of vertices $\boD=((u_{1},v_{1}),\dots$, $(u_{4k},v_{4k}))$ where $(u_{j},v_{j})$ is the $j$-th pair in $\boD$, let $M_A(\boD)$ be a perfect matching of $G_{\boD}$ defined as
\begin{equation*}
    M_A(\boD)=\{e=(u_{j},v_{2k+j},\mbox{red}):1\le j\le 2k\}\cup \{e=(u_{2k+j},v_{j},\mbox{red}):1\le j\le 2k\}.
\end{equation*}
In each block $\boA_i$, the matching is fixed to be $M_A(\boA_i)$. And $L_{\boA_i}$, the labeling restricted on $\boA_i$, can be any labeling consistent with $M_A(\boA_i)$.

Similarly, given a sequence of $4k$ pairs of vertices $\boD=((u_{1},v_{1}),\ldots,(u_{4k},v_{4k}))$, let $L_B(\boD)$ be a labeling of $G_{\boD}$ defined as follows.
\begin{itemize}
    \item For every $j\in [1,k]\cup [2k+1,3k]$, $L_B(\boD)(u_{j})=L_B(\boD)(v_{j})=1$.
    \item For every $j\in [k+1,2k]\cup [3k+1,4k]$, $L_B(\boD)(u_{j})=L_B(\boD)(v_{j})=0$.
\end{itemize}
Note that $M_A(\boD)$ is consistent with $L_B(\boD)$.

In each block $\boB_i$ the labeling is fixed to be $L_B(\boB_i)$. And $M_{\boB_i}$, the matching restricted on $\boB_i$, can be any perfect matching consistent with $L_B(\boB_i)$. 

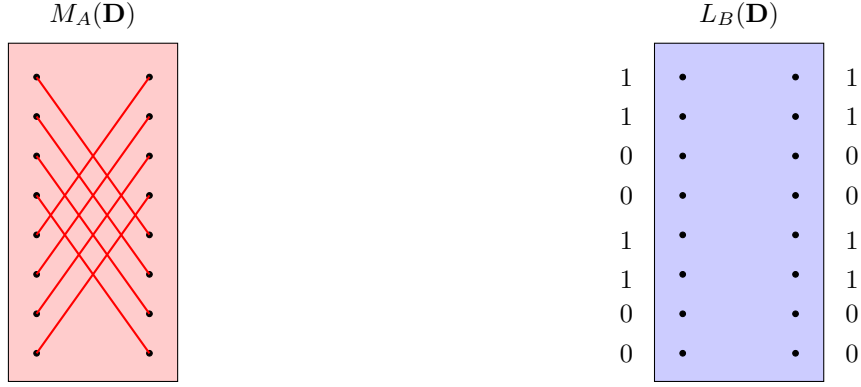
\begin{figure}
    \begin{minipage}[t][5.0cm][t]{0.49\textwidth}
    \begin{center}
    \begin{tikzpicture}[scale=0.75]
    \draw[draw=black, fill=red!20] (0, 0) rectangle ++(3, 6);
    \foreach \x in {0,...,7}
        {
        \draw [fill=black] (0.5, 0.5 + 0.7*\x) circle [radius = 0.05];
        \draw [fill=black] (2.5, 0.5 + 0.7*\x) circle [radius = 0.05];
        }
    \foreach \x in {0,...,3}
        {
        \draw [color=red, thick] (0.5, 0.5 + 0.7*\x) -- (2.5, 3.3 + 0.7*\x);
        \draw [color=red, thick] (0.5, 3.3 + 0.7*\x) -- (2.5, 0.5 + 0.7*\x);
        }
    \node at (1.5, 6.5) {$M_A(\boD)$};
    \end{tikzpicture}
    \end{center}
    \end{minipage}
    \begin{minipage}[t][5.0cm][t]{0.49\textwidth}
    \begin{center}
    \begin{tikzpicture}[scale=0.75]
    \draw[draw=black, fill=blue!20] (0, 0) rectangle ++(3, 6);
    \foreach \x in {0,...,7}
        {
        \draw [fill=black] (0.5, 0.5 + 0.7*\x) circle [radius = 0.05];
        \draw [fill=black] (2.5, 0.5 + 0.7*\x) circle [radius = 0.05];
        }
    \node at (-0.5, 5.4) {$1$};
    \node at (-0.5, 4.7) {$1$};
    \node at (-0.5, 2.5) {$1$};
    \node at (-0.5, 1.8) {$1$};
    \node at (3.5, 5.4) {$1$};
    \node at (3.5, 4.7) {$1$};
    \node at (3.5, 2.5) {$1$};
    \node at (3.5, 1.8) {$1$};
    
    \node at (-0.5, 4) {$0$};
    \node at (-0.5, 3.3) {$0$};
    \node at (3.5, 4) {$0$};
    \node at (3.5, 3.3) {$0$};
    \node at (-0.5, 1.2) {$0$};
    \node at (-0.5, 0.5) {$0$};
    \node at (3.5, 1.2) {$0$};
    \node at (3.5, 0.5) {$0$};
    
    \node at (1.5, 6.5) {$L_B(\boD)$};
    
    \end{tikzpicture}
    \end{center}
    \end{minipage}
    \caption{The matching $M_A(\boD)$ and the labeling $L_B(\boD)$. $\Mall(T)$ are all matchings that are exactly $M_A(\boD)$ on all A-blocks and consistent with $L_B(\boD)$ on all B-blocks. $\Lall(T)$ are labelings that are exactly $L_B(\boD)$ on B-blocks and consistent with $M_A(\boD)$ on A-blocks.}
    \label{fig:MA_and_LB}
\end{figure}

Define
\begin{align*}
\Mall(T) := & \{M:M_{\boA_i}=M_A(\boA_i),\, M_{\boB_i}\text{ is consistent with } L_B(\boB_i) \, \forall 1\le i\le m, \, M_C\text{ has odd number of} \\ & \text{ red edges}\}
\end{align*}
and
\begin{equation*}
\Lall(T):=\{L: L_{\boB_i}=L_B(\boB_i), \, L_{\boA_i}\text{ is consistent with } M_A(\boA_i) \, \forall 1\le i\le m, \, |L_C^{-1}(1)|\text{ is odd}\}.
\end{equation*}
For any $M\in\Mall(T)$, by a parity argument, both $M_{\boA_i}$ and $M_{\boB_i}$ contain an even number of red edges, which implies $M\in \Mall$. Similarly, for any $L\in \Lall(T)$, both $|L_{\boA_i}^{-1}(1)|$ and $|L_{\boB_i}^{-1}(1)|$ are even, which implies $|L^{-1}(1)|$ is odd. Thus $\Mall(T)\subseteq \Mall,\Lall(T)\subseteq \Lall$.
Note that for any fixed partition $T$, $M\in \Mall(T)$, and $L\in \Lall(T)$, violations caused by $(M,L)$ only occur in $C$.

\subsection{Generating the distributions}\label{subsec:generate}
From now on, we fix $\cR=\cM\times \cL$ to be a rectangle disjoint from 1-violation entries.

To quantitatively characterize $\mu_3(\cR)$ and $\mu_{4k+3}(\cR)$, we need the following notation. Let $T=(A=\{\boA_1,\ldots,\boA_m\},C,B=\{\boB_1,\ldots,\boB_m\})$ be a partition, $M'$ be a perfect matching of $C$, define
\begin{equation*}
p_{\cM,T}(M'):=\Pr_{M\sim \Mall(T,M')}[M\in \cM]
\end{equation*}
where $\Mall(T,M'):=\{M\in \Mall(T):M_C=M'\}$, as the probability that a uniform matching $M$ among all matchings w.r.t. $T$ whose restriction on $C$ is $M'$ lies in $\cM$.
Similarly, for any labeling $L'$ of $C$, define
\begin{equation*}
    p_{\cL,T}(L'):=\Pr_{L\sim \Lall(T,L')}[L\in \cL]
\end{equation*}
where $\Lall(T,L'):=\{L\in \Lall(T):L_C=L'\}$, as the probability that a uniform labeling $L$ among all labelings w.r.t. $T$ whose restriction on $C$ is $L'$ lies in $\cL$.

To generate $\mu_3$, we first sample a uniform partition $T$. Then we uniformly select three pairs of vertices from $C$. We call this unordered set $H$. Afterwards, we order the remaining $4k$ pairs of vertices uniformly at random to get $\boD=\boD_1\sqcup \boD_2$, where $\boD_1$ is an ordered set of the first $2k$ pairs of vertices in $\boD$ and $\boD_2$ is an ordered set of the last $2k$ pairs. Since $T,H,\boD$ will appear very frequently in the rest of the section, we use $\cT$ to denote the set of all the potential triples $(T,H,\boD)$ generated at this stage.

Let $M_3(H)$ be a perfect matching of $H$ where vertices in each pair are connected by a red edge and $L_3(H)$ be a labeling of $H$ where all the vertices of $H$ in $U$ are labeled with $1$ and those in $V$ are labeled with $0$. Note that all the edges in $M_3(H)$ violate $L_3(H)$.

Then we can generate $(M,L)$ according to a product distribution. Sample $M\sim \Mall(T,M_3(H\cup\boD))$, and $L\sim \Mall(T,L_3(H\cup\boD))$ uniformly and independently, where $M_3(H\cup \boD):=M_3(H)\cup M_A(\boD)$ and $L_3(H\cup \boD):=L_3(H)\cup L_B(\boD)$. 
Any potential outcome $(M,L)$ is a $3$-violation pair since the violations only occur in $H$.

Quantitatively, $\mu_3(\cR)$ can be expressed as
\begin{equation} \label{eqn:mu3}
    \mu_3(\cR)=\bE_T[\bE_{H,\boD}[p_{\cM,T}(M_3(H\cup \boD))\cdot p_{\cL,T}(L_3(H\cup \boD))]].
\end{equation}

On the other hand, to generate $\mu_{4k+3}$, we first sample a uniform partition $T$. Then we uniformly order the pairs of vertices in $C$. Let $\boC=\boF\sqcup \boD_2$ denote the obtained ordered set, where $\boF$ is the ordered set of the first $2k+3$ pairs of vertices and $\boD_2$ is the ordered set of the last $2k$ pairs. 

Similar to what we did for $\mu_3$, we can then sample $(M,L)$ according to a product distribution. Let $\boC=((u_{-2},v_{-2}),\ldots$, $(u_{4k},v_{4k}))$, where $(u_{j},v_{j})$ is the $(j+3)$-th pair of $\boC$. Let $M_{4k+3}(\boC)$ be a perfect matching of $\boC$ defined as
\begin{equation*}
    M_{4k+3}(\boC):=\{e=(u_{j},v_{j},\mathrm{red}):-2\le j\le 4k\}.
\end{equation*}

We also let $L_{4k+3}(\boC)$ be a labeling of $\boC$, where
\begin{itemize}
    \item for every $-2\le j\le 2k$, $L_{4k+3}(\boC)(u_{j})=1$, $L_{4k+3}(\boC)(v_{j})=0$;
    \item for every $2k+1\le j\le 4k$, $L_{4k+3}(\boC)(u_{2k+j})=0$, $L_{4k+3}(\boC)(v_{2k+j})=1$.
\end{itemize}

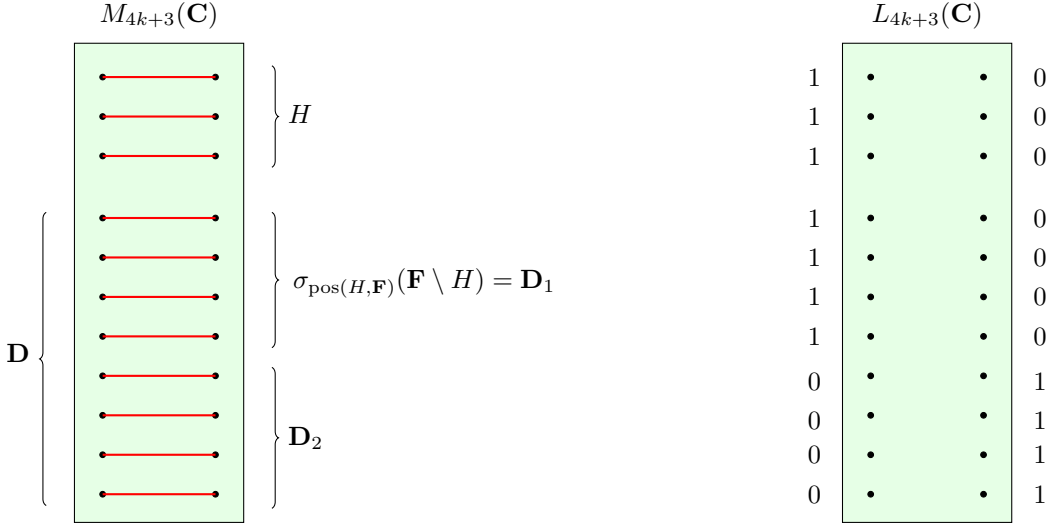
\begin{figure}
    \begin{minipage}[t][6.0cm][t]{0.49\textwidth}
    \begin{center}
    \begin{tikzpicture}[scale=0.75]
    \draw[draw=black, fill=green!10] (0, 0) rectangle ++(3, 8.5);
    \foreach \x in {0,...,7}
        {
        \draw [fill=black] (0.5, 0.5 + 0.7*\x) circle [radius = 0.05];
        \draw [fill=black] (2.5, 0.5 + 0.7*\x) circle [radius = 0.05];
        }
    \foreach \x in {0,...,2}
        {
        \draw [fill=black] (0.5, 6.5 + 0.7*\x) circle [radius = 0.05];
        \draw [fill=black] (2.5, 6.5 + 0.7*\x) circle [radius = 0.05];
        }
    \foreach \x in {0,...,7}
        {
        \draw [color=red, thick] (0.5, 0.5 + 0.7*\x) -- (2.5, 0.5 + 0.7*\x);
        }
    \foreach \x in {0,...,2}
        {
        \draw [color=red, thick] (0.5, 6.5 + 0.7*\x) -- (2.5, 6.5 + 0.7*\x);
        }
    \node at (1.5, 9) {$M_{4k+3}(\boC)$};
    
    \draw [decorate,
    decoration = {brace,mirror}] (3.5, 6.3) --  (3.5,8.1);
    \node at (4, 7.25) {$H$};
    \draw [decorate,
    decoration = {brace,mirror}] (3.5, 3.1) --  (3.5, 5.5);
    \node at (6.2, 4.25) {$\sigma_{\pos(H,\boF)}(\boF \setminus H) = \boD_1$};
    \draw [decorate,
    decoration = {brace,mirror}] (3.5, 0.25) --  (3.5, 2.75);
    \node at (4.1, 1.5) {$\boD_2$};
    
    \draw [decorate,
    decoration = {brace}] (-0.5, 0.3) --  (-0.5, 5.5);
    \node at (-1, 3) {$\boD$};
    
    \end{tikzpicture}
    \end{center}
    \end{minipage}
    \begin{minipage}[t][6.0cm][t]{0.49\textwidth}
    \begin{center}
    \begin{tikzpicture}[scale=0.75]
    \draw[draw=black, fill=green!10] (0, 0) rectangle ++(3, 8.5);
    \foreach \x in {0,...,7}
        {
        \draw [fill=black] (0.5, 0.5 + 0.7*\x) circle [radius = 0.05];
        \draw [fill=black] (2.5, 0.5 + 0.7*\x) circle [radius = 0.05];
        }
    \foreach \x in {0,...,2}
        {
        \draw [fill=black] (0.5, 6.5 + 0.7*\x) circle [radius = 0.05];
        \draw [fill=black] (2.5, 6.5 + 0.7*\x) circle [radius = 0.05];
        }
    \node at (-0.5, 5.4) {$1$};
    \node at (-0.5, 4.7) {$1$};
    \node at (-0.5, 2.5) {$0$};
    \node at (-0.5, 1.8) {$0$};
    \node at (3.5, 5.4) {$0$};
    \node at (3.5, 4.7) {$0$};
    \node at (3.5, 2.5) {$1$};
    \node at (3.5, 1.8) {$1$};
    
    \node at (-0.5, 4) {$1$};
    \node at (-0.5, 3.3) {$1$};
    \node at (3.5, 4) {$0$};
    \node at (3.5, 3.3) {$0$};
    \node at (-0.5, 1.2) {$0$};
    \node at (-0.5, 0.5) {$0$};
    \node at (3.5, 1.2) {$1$};
    \node at (3.5, 0.5) {$1$};
    
    \foreach \x in {0,...,2}
        {
        \node at (-0.5, 6.5 + 0.7*\x) {$1$};
        \node at (3.5, 6.5 + 0.7*\x) {$0$};
        }
    
    \node at (1.5, 9) {$L_{4k+3}(\boC)$};
    
    \end{tikzpicture}
    \end{center}
    \end{minipage}
    \caption{The matching $M_{4k+3}(\boC)$ and the labeling $L_{4k+3}(\boC)$. We also show the breakdown of $\boC = \boF \cup \boD_2$ in the triple $(T, H, \boD)$.}
    \label{fig:part-c}
\end{figure}

Then we sample $M\sim \Mall(T,M_{4k+3}(\boC))$ and $L\sim \Lall(T,L_{4k+3}(\boC))$ uniformly and independently. Any potential outcome $(M,L)$ is a $(4k+3)$-violation pair since all the edges in $M_{\boC}$ violate the labeling and all the matching edges in $A$ and $B$ are consistent with the labeling.

$\mu_{4k+3}(\cR)$ can also be quantitatively expressed as
\begin{equation*}
\mu_{4k+3}(\cR)=\bE_T[\bE_{\boC=\boF \sqcup \boD_2}[p_{\cM,T}(M_{4k+3}(\boC))\cdot p_{\cL,T}(L_{4k+3}(\boC))]].
\end{equation*}

\subsection{Classification of Triples}\label{subsec:classes}
In this section, we classify triples $(T,H,\boD)\in\cT$ into three categories: good, bad, or small w.r.t. the rectangle $\cR=\cM\times \cL$. From now on, we fix $0 < \epsilon < \frac{1}{3}$ to be a constant.

Given a triple $(T,H,\boD)\in \cT$, recall that $T=(A=\{\boA_1,\ldots,\boA_m\},C,B=\{\boB_1,\ldots,\boB_m\})$ is a partition, $\boD$ is a sequence of $4k$ pairs of vertices in $C$, and $H$ is the set of the remaining $3$ pairs in $C$. We say $(T,H,\boD)$ is $\cL$-good if the distribution of $L_{\boD}$ obtained by sampling a uniform $L\in \cL$ that satisfies $L_H=L_3(H)$ is close to a uniform distribution over all labelings consistent with $M_A(\boD)$. Below is the formal definition.
\begin{Definition}[$\cL$-good]
A triple $(T,H,\boD)\in \cT$ is said to be $\cL$-good if for any labeling of $\boD$ consistent with $M_A(\boD)$, denoted $L'$,
\begin{equation*}
    0<\frac{1}{1+\epsilon} p_{\cL,T}(L_3(H)\cup L')\le p_{\cL,T}(L_3(H\cup \boD))\le (1+\epsilon)p_{\cL,T}(L_3(H)\cup L').
\end{equation*}
\end{Definition}
Likewise, we say a triple $(T,H,\boD)$ is $\cM$-good if  the distribution of $M_{\boD}$ obtained by sampling a uniform $M\in \cM$ that satisfies $M_H=M_3(H)$ is close to a uniform distribution over all matchings consistent with $L_B(\boD)$.
\begin{Definition}[$\cM$-good]
A triple $(T,H,\boD)\in \cT$ is said to be $\cM$-good if for any perfect matching of $\boD$ consistent with $L_B(\boD)$, denoted $M'$,
\begin{equation*}
    0<\frac{1}{1+\epsilon} p_{\cM,T}(M_3(H)\cup M')\le p_{\cM,T}(M_3(H\cup \boD))\le (1+\epsilon)p_{\cM,T}(M_3(H)\cup M').
\end{equation*}
\end{Definition}

We say $(T,H,\boD)$ is good if it is both $\cM$-good and $\cL$-good.

If a triple $(T,H,\boD)$ is not good, it can be either bad or small. We say a triple is small if after fixing $T$, $H$, and $\boD$, the intersection of $\cR$ and the product distribution is exponentially small. A triple is bad if it is neither good nor small. Below is the formal definition.
\begin{Definition}
A triple $(T,H,\boD)\in \cT$ is small if either $p_{\cM,T}(M_3(H\cup \boD))\le 2^{-\delta m}$ or $p_{\cL,T}(L_3(H\cup \boD))\le 2^{-\delta m}$ for some constant $\delta:=\delta(k)>0$, which will be determined later. If a triple is neither good or small, then it is bad.
\end{Definition}

Let $\eG(T,H,\boD)$, $\eS(T,H,\boD)$, and $\eB(T,H,\boD)$ be the indicator functions of whether $(T,H,\boD)$ is good, small, or bad, respectively. We will get the following bounds for each type of triple of the right-hand side of \eqref{eqn:mu3}:

\begin{align}
    &\bE_T[\bE_{H,\boD}[\eG(T,H,\boD)\cdot p_{\cM,T}(M_3(H\cup \boD))\cdot p_{\cL,T}(L_3(H\cup \boD))]]\le \frac{20}{k^2}\mu_{4k+3}(\cR), \label{eqn:good} \\
    &\bE_T[\bE_{H,\boD}[\eS(T,H,\boD)\cdot p_{\cM,T}(M_3(H\cup \boD))\cdot p_{\cL,T}(L_3(H\cup \boD))]]\le 2^{-\delta m},
    \label{eqn:small} \\
    &\bE_T[\bE_{H,\boD}[\eB(T,H,\boD)\cdot p_{\cM,T}(M_3(H\cup \boD))\cdot p_{\cL,T}(L_3(H\cup \boD))]]\le \epsilon\cdot \mu_3(\cR)+2^{-\eta m}, \label{eqn:bad}
\end{align}
where $\eta:=\eta(k)>0$ is some constant which will be determined later.
Summing up \eqref{eqn:good},\eqref{eqn:small},and \eqref{eqn:bad}, we obtain
\begin{equation*}
    \mu_3(R)\le \frac{20}{k^2}\mu_{4k+3}(\cR)+2^{-\delta m}+\epsilon\cdot \mu_3(\cR)+2^{-\eta m}.
\end{equation*}
By choosing $\gamma=\min\{\delta,\eta\}/2$, we directly obtain \Cref{lem:over_cover}.

In fact, \eqref{eqn:small} can be proved in one line. For any triple $(T,H,\boD)$, the product in the expectation is upper bounded by $2^{-\delta m}$. Note that the expectation is always no larger than the maximum.

In the rest of this section we will prove \eqref{eqn:good} and \eqref{eqn:bad}.

\subsection{Contributions of Good Triples}\label{subsec:good}

In this section, we prove~\eqref{eqn:good}. That is, 

\[
\bE_T[\bE_{H,\boD}[\eG(T,H,\boD)\cdot p_{\cM,T}(M_3(H\cup \boD))\cdot p_{\cL,T}(L_3(H\cup \boD))]] \leq \frac{20}{k^2}\mu_{4k+3}(\cR)\,.
\tag{\ref{eqn:good}}\]
In the above expectation, we first select uniformly at random a partition $T=(A=\{\boA_1,\ldots,\boA_m\},C,B=\{\boB_1,\ldots,\boB_m\})$. 
Then $H$ is selected to be  three uniformly at random pairs of vertices from $C$, and finally we order the remaining $4k$ vertices uniformly at random to get $\boD$. 

We describe an alternative but equivalent viewpoint of selecting $T,H,\boD$. While this alternative viewpoint is more cumbersome, it will allow us to relate the above expectation to $\mu_{4k+3}$.  In this description we will use a family $\{\sigma_{\pos(H, \boF)}\}_{\pos(H, \boF) \in \binom{[2k+3]}{3}}$ of permutations. We later show, using the probabilistic method, that there is a family with nice properties that allows us to prove~\eqref{eqn:good} (see \Cref{lem:goodpermutations}).  The alternative way of selecting $T,H, \boD$ proceeds as follows. 
As before,  we first uniformly sample the partition $T$. 
But now we first uniformly at random permute all the $4k+3$ vertices in $C$ to get $\boC=\boF\sqcup \boD_2$, where we let $\boF$ denote the ordered set of the first $2k+3$ pairs of vertices and $\boD_2$ denote the ordered set of the remaining $2k$ pairs. 
Then we uniformly sample 3 pairs from $\boF$. 
Call this unordered set $H$. 
We let $\pos(H,\boF)\in \binom{[2k+3]}{3}$ denote the set of three indices representing the positions of $H$ in $\boF$.  
Then we apply  permutation $\sigma_{\pos(H,\boF)}$ to $\boF\setminus H$ to get $\boD_1=\sigma_{\pos(H,\boF)}(\boF\setminus H)$, where  $\boF\setminus H$ is an ordered set obtained by removing $H$ from $\boF$ without changing the order of the other pairs. 
Finally, we concatenate $\boD_2$ to the end of $\boD_1$ and get $\boD=\boD_1\sqcup \boD_2$. 
We remark that any three pairs of $\boF$ is a uniformly random sample of three pairs of $C$ due to the random permutation of $C$. 
So $H$ is three uniformly random pairs of $C$.  
Similarly, we have that $\boD$ is a uniformly random permutation of the remaining $4k$ pairs in $C$ no matter what $\sigma_{\pos(H, \boF)}$ is because the pairs in $\boF \setminus H \cup \boD_2$ are a random permutation of the $4k$ pairs in $C \setminus H$ (independent of $\sigma_{\pos(H, \boF)}$). 
It follows that this alternative viewpoint is equivalent and so we can express the left-hand side of \eqref{eqn:good} by
\begin{align*}
    &\bE_T[\bE_{H,\boD}[\eG(T,H,\boD)\cdot p_{\cM,T}(M_3(H\cup \boD))\cdot p_{\cL,T}(L_3(H\cup \boD))]]\\
    =&\bE_T[\bE_{\boC=\boF\sqcup\boD_2}[\bE_{H\sim \binom{\boF}{3},\boD_1=\sigma_{\pos(H,\boF)}(\boF\setminus H),\boD=\boD_1\sqcup
    \boD_2}[\eG(T,H,\boD)\cdot p_{\cM,T}(M_3(H\cup \boD))\cdot p_{\cL,T}(L_3(H\cup \boD))]]].
\end{align*}
We further have that this expectation is upper bounded by
\begin{align*}
    & 2\cdot \bE_T[\bE_{\boC=\boF\sqcup \boD_2}[p_{\cM,T}(M_{4k+3}(\boC))\cdot p_{\cL,T}(L_{4k+3}(\boC))\cdot \bE_{H\sim \binom{\boF}{3},\boD_1=\sigma_{\pos(H,\boF)}(\boF\setminus H),\boD=\boD_1\sqcup \boD_2}[\eG(T,H,\boD)]]].
\end{align*}
This follows from the following observation: when $(T,H,\boD)$ is good, $p_{\cM,T}(M_3(H\cup \boD))\le (1+\epsilon)p_{\cM,T}(M_{4k+3}(\boC))$ since $(M_{4k+3}(\boC))_H=M_3(H)$ and $(M_{4k+3}(\boC))_{\boD}$ is consistent with $L_B(\boD)$. Similarly, $p_{\cL,T}(L_3(H\cup \boD))\le (1+\epsilon)p_{\cL,T}(L_{4k+3}(\boC))$ since $(L_{4k+3}(\boC))_H=L_3(H)$ and $(L_{4k+3}(\boC))_{\boD}$ is consistent with $M_A(\boD)$.

Now recall that $\mu_{4k+3}(\cR)=\bE_T[\bE_{\boC=\boF \sqcup \boD_2}[p_{\cM,T}(M_{4k+3}(\boC))\cdot p_{\cL,T}(L_{4k+3}(\boC))]]$ and so by combining the above inequalities,
\begin{align*}
    & \bE_T[\bE_{H,\boD}[\eG(T,H,\boD)\cdot p_{\cM,T}(M_3(H\cup \boD))\cdot p_{\cL,T}(L_3(H\cup \boD))]]\\
    \le& 2\mu_{4k+3}(\cR)\cdot \max_{T,\boC=\boF\sqcup \boD_2}\{\bE_{H\sim \binom{\boF}{3},\boD_1=\sigma_{\pos(H,\boF)}(\boF\setminus H),\boD=\boD_1\sqcup \boD_2}[\eG(T,H,\boD)]\}.
\end{align*}
The remainder of this section will be devoted to showing  that for any partition $T$ and any ordering $\boC=\boF\sqcup \boD_2$,
\begin{equation}\label{eqn:prob_good}
    \bE_{H\sim \binom{\boF}{3},\boD_1=\sigma_{\pos(H,\boF)}(\boF\setminus H),\boD=\boD_1\sqcup \boD_2}[\eG(T,H,\boD)]\le \frac{10}{k^2}.
\end{equation}
Then \eqref{eqn:good} directly follows.

Our main technical ingredient is the following. Let $H$ and $H'$ be two sets intersecting on at most one pair of vertices. Then under a technical assumption, at least one of $(T,H,\boD)$ and $(T,H',\boD')$ is not good. We remark that this technical assumption is one of the key differences between our proof and that of Rothvoss's proof for general matchings. Indeed, the analogous statement in that paper  always guarantees that at least one of $(T,H,\boD)$ and $(T,H',\boD')$ is not good. Moreover, this is also the reason why we use the family of permutations $\{\sigma_{\pos(H, \boF)}\}_{\pos(H, \boF) \in \binom{[2k+3]}{3}}$ with certain properties.

\begin{lemma}\label{lem:not_both_good}
Let $(T,H,\boD=\boD_1\sqcup\boD_2)$ and $(T',H',\boD'=\boD'_1\sqcup \boD'_2)$ be two triples such that $T=T'$, $\boD_2=\boD'_2$, and $|H\cap H'|\le 1$. If $(u_1,v_1)$ and $(u_2,v_2)$ are two pairs belonging to $H'\setminus H$ and $L_B(\boD)(v_1)\ne L_B(\boD)(v_2)$, then $(T,H,\boD)$ and $(T',H',\boD')$ cannot both be good.
\end{lemma}

\begin{proof}

The goal of the proof is to find a 1-violation pair $(M,L)\in \cR=\cM\times \cL$ that contradicts the assumption that $\cR$ contains no such entry.
This is conditioned on $(T,H,\boD)$ and $(T,H',\boD')$ both being good and the existence of two pairs $(u_1, v_1), (u_2, v_2) \in H'\setminus H$ such that $L_B(\boD)(v_1) \neq L_B(\boD)(v_2)$.
We now give five conditions that are necessary and sufficient to construct an auxiliary pair $(M', L')$ that proves the existence of the 1-violation pair.


{Recall the definitions of $C$, $H$, and $D$ as shown in Figure~\ref{fig:part-c}, of $A$ and $B$ as shown in Figure~\ref{fig:partition}, and of $M_A(\boD),L_B(\boD)$ as shown in Figure~\ref{fig:MA_and_LB}.}
To prove the existence of a $1$-violation pair $(M,L)\in \cR$, it suffices to find a perfect matching $M'$ and a labeling $L'$ of $C$ such that
\begin{itemize}
    \item[(i)] $L'_{H'}=L_3(H')$
    \item[(ii)]
    $L'_{\boD'}$ is consistent with $M_A(\boD')$
    \item[(iii)] $M'_{H}=M_3(H)$
    \item[(iv)]
    $M'_{\boD}$ is consistent with $L_B(\boD)$
    \item[(v)]
    There is exactly one edge in $M'$ that violates $L'$.
\end{itemize}

Now we explain more explicitly how the goodness assumption together with these five conditions imply we can find a 1-violation. 
If we can find such a pair $(M',L')$, then by (i), (ii), and that $(T,H',\boD')$ is $\cL$-good, $p_{\cL,T}(L')>0$, which implies that $\cL\cap \Lall(T,L')\ne \emptyset$. Let $L$ be an arbitrary labeling in $\cL\cap \Lall(T,L')$. 
Similarly, by (iii), (iv), and that $(T,H,\boD)$ is $\cM$-good, $p_{\cM,T}(M')>0$. 
This implies $\cM\cap \Mall(T,M')\ne \emptyset$. Let $M$ be an arbitrary matching in $\cM\cap \Mall(T,M')$. 
Since $M\in \Mall(T)$ and $L\in \Lall(T)$ are sampled from the same partition $T$, $L$ is consistent with $M$ on $A$ and $B$. Together with (v), we conclude that $(M,L)$ is a $1$-violation pair and $(M,L) \in \cR$ because $M \in \cM$ and $L \in \cL$.

Next, we exhibit a pair $(M',L')$ that satisfies (i)-(v). We first give a simple construction of $L'$:
\begin{equation}
    L'(v)=\left\{
    \begin{matrix}
    1 & v \notin C \setminus (H' \cap V) \\
    0 & v \in H' \cap V
    \end{matrix}
    \right..
\end{equation}


Let us verify that $L'$ satisfies (i) and (ii). Recall that $L_3(H)$ assigns $1$s to vertices in $H'\cap U$ and $0$s to vertices in $H'\cap V$. Hence $L'_{H'}=L_3(H')$, i.e. (i) holds. Moreover, since $L'_{\boD'}$ is the constant-1 function and $M_A(\boD')$ contains only red edges, $L'_{\boD'}$ is consistent with $M_A(\boD')$, i.e. (ii) holds.

Next, we construct $M'$ as follows:
\begin{equation*}
    M'=\{(u,v,\text{red}),(u,v)\in C\setminus \{(u_1,v_1),(u_2,v_2)\}\}\cup \{(u_1,v_2,\text{blue}),(u_2,v_1,\text{blue})\}.
\end{equation*}

Now let us verify that $M'$ satisfies (iii) and  (iv). Recall that $M_3(H)$ is the set of three red edges that connect pairs in $H$. Since $H\subseteq C\setminus \{(u_1,v_1),(u_2,v_2)\}$, $M'_H=M_3(H)$, i.e. (iii) holds. For (iv), note that for any $(u,v)\in \boD$, $L_B(\boD)(u)=L_B(\boD)(v)$, thus all the red edges in $M'_\boD$ are consistent w.r.t. $L_B(\boD)$. 
Moreover, by the assumption of the lemma $L_B(\boD)(v_1)\ne L_B(\boD)(v_2)$, which implies $L_B(\boD)(u_1)\ne L_B(\boD)(v_2),L_B(\boD)(u_2)\ne L_B(\boD)(v_1)$ since $L_B(\boD)(u_1) = L_B(\boD)(v_1)$ and $L_B(\boD)(u_2) = L_B(\boD)(v_2)$. 
Thus the two blue edges are also consistent. Hence (iv) holds. 

Finally, let us verify that $M'$ along with $L'$ satisfies (v). Since for any $(u,v)\in C\setminus H'$, $L'(u)=L'(v)=1$, the $4k$ red edges in $M'$ that connect pairs in $C\setminus H'$ are consistent w.r.t. $L'$. 
Let $(u_3,v_3)$ denote the other pair in $H'$ other than $(u_1,v_1)$ and $(u_2,v_2)$. Note that $L'(u_1)=L'(u_2)=L'(u_3)=1,L'(v_1)=L'(v_2)=L'(v_3)=0$. 
Thus the two blue edges $(u_1,v_2),(u_2,v_1)$ are consistent w.r.t. $L'$ while the red edge $(u_3,v_3)$ is the only violating edge. Therefore, (v) holds.

In summary, we can find $(M',L')$ that satisfies (i)-(v), which implies the existence of $1$-violation entry $(M,L)\in \cR$, leading to a contradiction.

\begin{figure}
    \begin{minipage}[t][3.5cm][t]{0.49\textwidth}
    \begin{center}
    \begin{tikzpicture}
    
    \foreach \x in {0,...,10}
    {
        \draw [fill=black] (0, 0 + 0.5*\x) circle [radius = 0.05];
        \draw [fill=black] (1.5, 0 + 0.5*\x) circle [radius = 0.05];
        \node at (-0.4, 0 + 0.5*\x) {$\color{red}1$};
    }
    
    \foreach \x in {0,...,2}
    {
        \node at (-0.6, \x*0.5) {$\color{blue}0$};
        \node at (2.1, \x*0.5) {$\color{blue}0$};
        \node at (1.9, \x*0.5) {$\color{red}1$};
        \node at (-0.6, \x*0.5+1.5) {$\color{blue}1$};
        \node at (2.1, \x*0.5+1.5) {$\color{blue}1$};
        \node at (1.9, \x*0.5+1.5) {$\color{red}1$};
    }
    
    \node at (1.9, 4.5) {$\color{red} 1$};
    \node at (1.9, 5) {$\color{red} 1$};
    
    \node at (-0.6, 3) {$\color{blue} 0$};
    \node at (2.1, 3) {$\color{blue} 0$};
    \node at (-0.6, 3.5) {$\color{blue} 1$};
    \node at (2.1, 3.5) {$\color{blue} 1$};
    
    \foreach \x in {0,...,2}
    {
        
        \node at (1.9, 3 + 0.5*\x) {$\color{red}0$};
    }
    
    \draw [color=red] (0, 4.5) -- (1.5, 4.5);
    \draw [color=red] (0, 5) -- (1.5, 5);
    \draw[draw=red, decorate, decoration=snake] (0, 4) -- (1.5, 4);

    \node at (-1, 3.5) {$u_1$};
    \node at (-1, 3) {$u_2$};
    \node at (2.5, 3.5) {$v_1$};
    \node at (2.5, 3) {$v_2$};
    \node at (-1, 4) {$u_3$};
    \node at (2.5, 4) {$v_3$};
    
    
    \draw [color=blue] (0,3.5) -- (1.5, 3);
    \draw [color=blue] (1.5,3.5) -- (0, 3);
    \foreach \x in {0,...,5}
    {
        \draw [color=red] (0,0.5*\x) -- (1.5,0.5*\x);
    }
    
    \draw [decorate,
    decoration = {brace,mirror}] (3, 3.8) --  (3,5.25);
    \node at (3.5, 4.5) {$H$};
    \draw [decorate,
    decoration = {brace,mirror}] (3.2, 2.8) --  (3.2, 4.1);
    \node at (3.7, 3.5) {$H'$};
    
    \end{tikzpicture}
    \end{center}
    \end{minipage}
    \begin{minipage}[t][3.5cm][t]{0.49\textwidth}
    \begin{center}
    \begin{tikzpicture}
    
    \foreach \x in {0,...,10}
    {
        \draw [fill=black] (0, 0 + 0.5*\x) circle [radius = 0.05];
        \draw [fill=black] (1.5, 0 + 0.5*\x) circle [radius = 0.05];
        \node at (-0.4, 0 + 0.5*\x) {$\color{red}1$};
    }
    
    \foreach \x in {0,...,2}
    {
        \node at (-0.6, \x*0.5) {$\color{blue}0$};
        \node at (2.1, \x*0.5) {$\color{blue}0$};
        \node at (-0.6, \x*0.5+1.5) {$\color{blue}1$};
        \node at (2.1, \x*0.5+1.5) {$\color{blue}1$};
    }
    
    \foreach \x in {0,...,4}
    {
        \node at (1.9,\x*0.5) {$\color{red} 1$};
        \draw [color=red] (0,0.5*\x) -- (1.5,0.5*\x);
    }
    \node at (1.9,4) {$\color{red} 1$};
    
    \node at (1.9, 4.5) {$\color{red} 1$};
    \node at (1.9, 5) {$\color{red} 1$};
    
    \node at (-0.6, 3) {$\color{blue} 0$};
    \node at (2.1, 3) {$\color{blue} 0$};
    \node at (-0.6, 3.5) {$\color{blue} 1$};
    \node at (2.1, 3.5) {$\color{blue} 1$};
    
    \foreach \x in {0,...,2}
    {
        
        \node at (1.9, 2.5 + 0.5*\x) {$\color{red}0$};
    }
    
    \draw [color=red] (0, 4) -- (1.5, 4);
    \draw [color=red] (0, 4.5) -- (1.5, 4.5);
    \draw [color=red] (0, 5) -- (1.5, 5);
    \draw[draw=red, decorate, decoration=snake] (0, 2.5) -- (1.5, 2.5);

    \node at (-1, 3.5) {$u_1$};
    \node at (-1, 3) {$u_2$};
    \node at (2.5, 3.5) {$v_1$};
    \node at (2.5, 3) {$v_2$};
    \node at (-1, 2.5) {$u_3$};
    \node at (2.5, 2.5) {$v_3$};
    
    
    \draw [color=blue] (0,3.5) -- (1.5, 3);
    \draw [color=blue] (1.5,3.5) -- (0, 3);
    
    \draw [decorate,
    decoration = {brace,mirror}] (3, 3.8) --  (3,5.25);
    \node at (3.5, 4.5) {$H$};
    \draw [decorate,
    decoration = {brace,mirror}] (3, 2.3) --  (3, 3.6);
    \node at (3.5, 3) {$H'$};
    
    \end{tikzpicture}
    \end{center}
    \end{minipage}
    \caption{The above two figures show how $M'$ and $L'$ are constructed when $k=2$. The left one shows the case when $|H\cap H'|=1$ while the right one shows the case when $|H\cap H'|=0$. The red labeling is $L'$ and the blue one is $L_B(\boD)$. The squiggly line is the only violating edge in both figures.}
    \label{fig:1-intersection}
\end{figure}
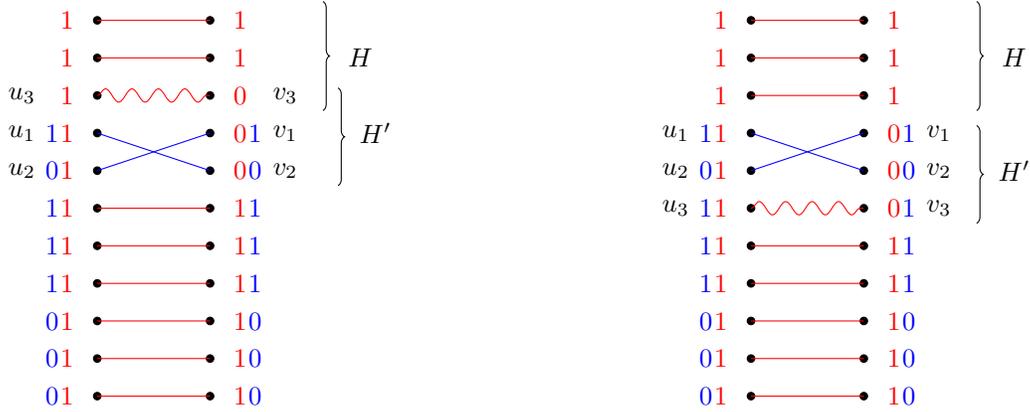

\end{proof}

Finally, we show that there exists a set of permutations $\{\sigma_t:t\in [2k+3]^3\}$ such that \eqref{eqn:prob_good} holds. The proof is by the probabilistic method.

\begin{lemma}
    There is a family $\{\sigma_{t} : t \in \binom{[2k+3]}{3}\}$ of bijections $\sigma_t : [2k+3]\setminus t \rightarrow [2k]$   such that the following holds.  For any subset $S$ of $\binom{[2k+3]}{3}$ with $|S| \geq 10k$,   there are $t, t' \in S$ and two elements $v_1, v_2 \in t' \setminus t$ satisfying $1 \leq \sigma_{t}(v_1) \leq k < \sigma_{t}(v_2) \leq 2k$.
    \label{lem:goodpermutations}
\end{lemma}
Before proving this lemma, let us argue that it implies~\eqref{eqn:prob_good}. Let $\{\sigma_{\pos(H, \boF)} : \pos(H,\boF) \in \binom{[2k+3]}{3}\}$ be the family guaranteed by the lemma. Fix any partition $T$ and any ordering  $\boC=\boF\sqcup \boD_2$. Suppose toward a contradiction that
\[
    \bE_{H\sim \binom{\boF}{3},\boD_1=\sigma_{\pos(H,\boF)}(\boF\setminus H),\boD=\boD_1\sqcup \boD_2}[\eG(T,H,\boD)] >  \frac{10}{k^2}.
\]
As there are $\binom{2k+3}{3} \geq k^3$ many choices of $H$, this implies that there are at most $10k$ choices of $H$ such that $(T,H, \boD)$ is good.  By the above lemma, two of these choices, say $H$ and $H'$ with corresponding triples $(T, H, \boD)$ and $(T, H', \boD')$, must be such that there exist two pairs  $(u_1, v_1), (u_2, v_2) \in H' \setminus H$ satisfying  $1 \leq \sigma_{\pos(H,\boF)}(v_1) \leq k < \sigma_{\pos(H,\boF)}(v_2)  \leq2k $.   
By the definition of $L_B(\boD)$, this means that $L_B(\boD)(v_1) \neq L_B(\boD)(v_2)$. \Cref{lem:not_both_good} now gives the contradiction that not both $(T, H, \boD)$ and $(T,H', \boD')$ can be good.

It remains to prove that there is a family of ``good'' permutations.

\begin{proof}[Proof of \Cref{lem:goodpermutations}]
The proof is by the probabilistic method: we show that a random family is likely to satisfy the statement.

Given two  $t,t'\in \binom{[2k+3]}{3}$, let $E(t,t')$ denote the event that either $|t\cap t'|\ge 2$, or for any two elements $v_1$ and $v_2$ belonging to $t'\setminus t$, either $1\le \sigma_{t}(v_1),\sigma_{t}(v_2)\le k$ or $k+1\le \sigma_{t}(v_1),\sigma_{t}(v_2)\le 2k$. 
Note that if $E(t, t')$ is false then there must be $v_1, v_2 \in t'\setminus t$ such that $1\leq \sigma_{t}(v_1) \leq k < \sigma_{t}(v_2) \leq 2k$. 

Now fix a subset $S$ of $\binom{[2k+3]}{3}$ such that $|S| = 10k$.  We let $E(t,S)$ denote the event that $E(t,t')$ happens for all $t'\in S$ and $E(S)$ denote the event that $E(t,S)$ happens for all $t\in S$. We will prove that $\Pr[E(S)] \leq 2^{-\Omega(k^{4/3})}$. There are 
\[
    \binom{\binom{2k+3}{3}}{10k} = k^{O(k)}
\]
different subsets $S$ of cardinality $10k$. Thus by the union bound,  the event $E(S)$ is false for all subsets $S$ of cardinality $10k$ with probability $1-o_k(1)$. This implies the lemma. Indeed, a random family is likely to satisfy the statement of the lemma.  

It remains to show that  $E(S)$ happens with probability at most $2^{-\Omega(k^{4/3})}$ for a subset $S$ of cardinality $10k$. 
Each bijection $\sigma_{t}$ is chosen independently uniformly at random. 
Moreover, the only randomness that impacts the event $E(t, S)$ is the choice of $\sigma_{t}$. 
It follows that the events in $\{E(t, S) : t\in S\}$ are mutually independent. 
So to prove that $\Pr[E(S)] \leq 2^{-\Omega(k^{4/3})}$, it suffices  to show that $\Pr[E(t,S)]=2^{-\Omega(k^{1/3})}$ for $t \in S$.

There are at most $3k$ different triples $t'$ such that $|t\cap t'|\ge 2$. Since $E(t,t')$ is always true in this case, we can remove such $t'$ and obtain $S'\subseteq S$ such that $|S'|\geq 7k$ and every $t'\in S'$ satisfies $|t \cap t'| \leq 1$. Name the elements of $S' = \{t_1, t_2, \ldots, t_\ell\}$ where $\ell = |S'| \geq 7k$. 

We say that $t_i$ is special if $t_i\setminus (t_1\cup \cdots\cup t_{i-1})\ne \emptyset$. Let $s_i$ denote the number of special elements in $t_1,\ldots,t_{i}$. Then \[r_i:=|t_1\cup \cdots \cup t_{i}|\le 3s_i.\] 
For a special $t_i$, consider an arbitrary element $x$  in $t_i\setminus (t_1\cup \cdots \cup t_{i-1})$. If we condition on  the events $E(t,t_1)$, $E(t, t_2), \ldots,E(t, t_{i-1})$ happening, the probability that $1\le \sigma_{t}(x)\le k$ (or $k+1\le\sigma_{t}(x)\le 2k$) is at most $\frac{k}{2k-r_i}\le\frac{k}{2k-3s_i}$. 
Moreover, since a set of size $r$ can have at most $\binom{r}{3}$ different subsets of size $3$, we have $7k \leq \ell\le \binom{r_i}{3}\le \binom{3s_\ell}{3}$, which implies that $s_\ell=\Omega\left(k^{1/3}\right)$. Finally, we apply the chain rule and obtain that
\begin{align*}
\Pr[E(t,S)]&=\Pr[E(t,t_1)\land\cdots\land E(t,t_{\ell})]\\
&=\prod_{i=1}^{\ell} \Pr[E(t,t_i)|E(t,t_1)\land\cdots\land E(t,t_{i-1})]\\
&\le\prod_{1\le i\le \ell, i\text{ special}} \Pr[E(t,t_i)|E(t,t_1)\land\cdots\land E(t,t_{i-1})]\\
&\le \prod_{0\le i\le \min(s_{\ell}, k/4)} \frac{k}{2k-3i}\\
&=2^{-\Omega(k^{1/3})}.
\end{align*}

As aforementioned, this implies that $\Pr[E(S)] \leq 2^{-\Omega(k^{4/3})}$, which in turn implies the lemma by taking the union bound over all $S$ with $|S| = 10k$.
\end{proof}

\subsection{Contributions of Bad Triples}\label{subsec:bad}

Finally, let us bound the contribution of bad triples. We restate the desired inequality here.
\begin{equation*}
    \bE_T[\bE_{H,\boD}[\eB(T,H,\boD)\cdot p_{\cM,T}(M_3(H\cup \boD))\cdot p_{\cL,T}(L_3(H\cup \boD))]]\le \epsilon\cdot \mu_3(\cR)+2^{-\eta m}. \tag{\ref{eqn:bad}}
\end{equation*}

{We need the following lemma from~\cite{Rothvoss17}, which was implicitly stated in~\cite{RAZBOROV1992385} to prove the set disjointness lower bound, and has been widely used in various forms in other works as well.} The lemma is essentially based on an entropy counting argument.
\begin{lemma}[\cite{Rothvoss17}]\label{lem:biased}
For every $\alpha,\beta>0$ and $q\in \bN$, there is a constant $\delta:=\delta(\alpha,\beta,q)$ such that the following holds. Let $X_1,\ldots,X_m$ be sets of size $|X_i|\le q$, and denote $X:=X_1\times\cdots\times X_m$. Let $Y\subseteq X$ be a subset of size $|Y|\ge 2^{-\delta m}|X|$. We say an index $i\in [m]$ is $\alpha$-unbiased, if
\begin{equation*}
    \frac{1}{1+\alpha}\Pr_{y\in X}[y\in Y]\le \Pr_{y\in X}[y\in Y:y_i=j]\le (1+\alpha)\Pr_{y\in X}[y\in Y].
\end{equation*}
for every $j\in X_i$. Then at most $\beta m$ indices are $\alpha$-biased.
\end{lemma}

To prove \eqref{eqn:bad}, we use the same framework as that in Rothvoss's paper. The rough idea is as follows. First, we swap the order of the expectations in the left-hand side of \eqref{eqn:bad}. 
For each fixed $3$-violation entry $(M,L)$, we sample the uniform distribution over the set of all triples $(T,H,\boD)$ that generate $(M,L)$.
Let $J$ denote the set of indices $i\in [m]$ such that $L_{\boA_i}=L_B(\boA_i)$. Then we pick a uniformly random $i\in J\cup \{0\}$, and exchange $\boD$ with $\boA_j$ (we let $\boA_0:=\boD$).
Observe that we still obtain a uniform distribution over tuples that generate $(M,L)$.
If $(T,H,\boD)$ is bad, then it is not small, so by \Cref{lem:biased}, at most $\epsilon'm$ indices are biased. 
The ideal situation is that $|J|\ge \frac{\epsilon'}{\epsilon}{m}$ so that at most an $\epsilon$-fraction of triples are bad.
However, it is possible that $|J|$ is small, say $o(m)$. 
In that case, we cannot bound the contribution since we only have an $O(m)$ upper bound on the number of biased indices. 
To get around this, we need the notions of balanced matchings and balanced labelings. 

Given a partition $T$, we say a matching $M\in \Mall(T)$ is balanced w.r.t. $T$ if $M_{\boB_i}=M_A(\boB_i)$ for at least a $\frac{1}{2c_k}$-fraction of indices $i\in [m]$, where $c_k:=(4k)!$.
Let $\BalM(T):=\{M\in \Mall(T):|\{i\in [m]:M_{\boB_i}=M_A(\boB_i)\}|\ge  \frac{m}{2c_k}\}$ denote the set of all such matchings. 
Likewise, we say a labeling $L\in \Lall(T)$ is balanced w.r.t. $T$ if $L_{\boA_i}=L_B(\boA_i)$ for at least a $\frac{1}{2c_k}$-fraction of $i\in [m]$. 
Let $\BalL(T):=\{L\in \Lall(T):|\{i\in [m]:L_{\boA_i}=L_B(\boA_i)\}| \ge \frac{m}{2c_k}\}$ denote the set of all such labelings.
Note that the number of perfect matchings consistent with some particular labeling on a block of size $4k$ is exactly $c_k$ and the number of labelings consistent with some particular perfect matching on a block of size $4k$ is $2^{4k}<c_k$.
Recall the process that generates $\mu_3$. We first sample the triple $(T,H,\boD)$, and then we sample $(M,L)$ according to a product distribution: $M\sim \Mall(T,M_3(H\cup \boD)),L\sim \Lall(T,L_3(H\cup \boD))$. 
Observe that $M_{\boA_1},\ldots,M_{\boA_m},L_{\boB_1},\ldots,L_{\boB_m}$ are independent. Moreover, $\bE[|\{i\in [m]:L_{\boA_i}=L_B(\boA_i)\}|]>\bE[|\{i\in [m]:M_{\boB_i}=M_A(\boB_i)\}|]=\frac{m}{c_k}$. 
By the Chernoff bound and union bound, for sufficiently large $m$, we thus have
\begin{equation*}
\bE_{M\sim \Mall(T,M_3(H\cup \boD)),L\sim \Lall(T,L_3(H\cup \boD))}\left[\bone_{M\notin \BalM(T)\lor L\notin \BalL(T)}\right]
\le 2e^{-\frac{m}{8c_k}}<2^{-\eta m},
\end{equation*}
where $\eta:=\eta(k)=\frac{1}{16c_k}$.
In other words, for any fixed partition $T$, the contribution of unbalanced entries is tiny. Thus we only need to bound the contribution of balanced entries:
\begin{align}\label{eqn:ctrbtn_balanced}
&\bE_{T,H,\boD}\left[\bE_{M\sim \Mall(T,M_3(H\cup \boD)),L\sim \Lall(T,L_3(H\cup \boD))}\left[\eB(T,H,D)\cdot \bone_{(M,L)\in \cR \land M\in \BalM(T)\land L\in \BalL(T)}\right]\right] \nonumber\\
&\le \epsilon\cdot \mu_3(\cR).
\end{align}
By swapping the order of expectations, we get
\begin{equation}\label{eqn:swap_expect}
    \text{L.H.S. of }\eqref{eqn:ctrbtn_balanced}=\bE_{(M,L)\sim \mu_3}\left[\bone_{(M,L)\in \cR}\cdot \bE_{(T,H,\boD)\sim \cP(M,L)}[\eB(T,H,\boD)\cdot \bone_{M\in \BalM(T)}\cdot\bone_{L\in \BalL(T)}]\right],
\end{equation}
where $\cP(M,L)$ is defined as the set of triples $(T,H,\boD)$ that could generate $(M,L)$, i.e.
\begin{equation*}
\cP(M,L):=\{(T,H,\boD):M\in \Mall(T,M_3(H\cup \boD)),L\in \Lall(T,L_3(H\cup \boD))\}.    
\end{equation*}

The reason why the inner expectation of the right-hand side of \eqref{eqn:swap_expect} is taking the uniform distribution is given by Bayes' rule and the fact that we sample the triples $(T,H,\boD)$ with equal probability when we generate $\mu_3$:
\begin{equation*}
    \Pr[(T,H,\boD) \mid (M,L)]=\frac{\Pr[(T,H,\boD)]\Pr[(M,L)\mid (T,H,\boD)]}{\Pr[(M,L)]}=\frac{\Pr[(T,H,\boD)]}{ ((4k)!2^{4k})^m\mu_3(M,L)} .
\end{equation*}

To complete the proof of \eqref{eqn:bad}, it suffices to bound the proportion of bad triples for each pair $(M,L)$:
\begin{equation}\label{eqn:tbc_of_entry}
    \bE_{(T,H,\boD)\sim \cP(M,L)}[\eB(T,H,\boD)\cdot \bone_{M\in \BalM(T)}\cdot\bone_{L\in \BalL(T)}]\le \epsilon.
\end{equation}

Indeed, \eqref{eqn:ctrbtn_balanced} directly follows from \eqref{eqn:tbc_of_entry} since
\begin{equation*}
    \eqref{eqn:swap_expect}\le \bE_{(M,L)\sim \mu_3}\left[\bone_{(M,L)\in \cR}\cdot \epsilon\right]=\epsilon\cdot\mu_3(\cR).
\end{equation*}

Given a triple $(T,H,\boD)\in \cT$, we say $(T,H,\boD)$ is $\cM-$bad (resp. $\cL-$bad), if $(T,H,\boD)$ is neither small nor $\cM-$good (resp. $\cL-$good). Let $\cM-\eB(T,H,\boD)$ (resp. $\cL-\eB(T,H,\boD)$) denote its corresponding indicator function. Note that if $(T,H,\boD)$ is bad, it is either $\cM-$bad or $\cL-$bad. Then we have
\begin{align*}
    \text{L.H.S. of }\eqref{eqn:tbc_of_entry} &\le \bE_{(T,H,\boD)\sim \cP(M,L)} [\cM-\eB(T,H,\boD)\cdot \bone_{M\in \BalM(T)}] \\
    & +\bE_{(T,H,\boD)\sim \cP(M,L)} [\cL-\eB(T,H,\boD)\cdot \bone_{L\in \BalL(T)}].
\end{align*}

We will bound the contribution of $\cM$-bad triples and $\cL$-bad triples separately.
\begin{lemma}\label{lem:Lbad}
Let $(M,L)$ be a $3$-violation entry in the support of $\mu_3$. Then
\begin{equation*}
    \bE_{(T,H,\boD)\sim \cP(M,L)} [\cL-\eB(T,H,\boD)\cdot \bone_{L\in \BalL(T)}]\le \frac{\epsilon}{2}.
\end{equation*}
\end{lemma}
\begin{lemma}\label{lem:Mbad}
Let $(M,L)$ be a $3$-violation entry in the support of $\mu_3$. Then
\begin{equation*}
    \bE_{(T,H,\boD)\sim \cP(M,L)} [\cM-\eB(T,H,\boD)\cdot \bone_{M\in \BalM(T)}]\le \frac{\epsilon}{2}.
\end{equation*}
\end{lemma}
Here we will only present the proof for \Cref{lem:Lbad}. \Cref{lem:Mbad} can be proved analogously.

\begin{proof}[Proof of \Cref{lem:Lbad}]
We first sample a uniform triple $(T,H,\boD)\sim \cP(M,L)$. In order to show the desired bound, we introduce one more step, which yields the same uniform distribution. Let $J(T,L)=\{0\le i\le m:L_{\boA_i}=L_B(\boA_i)\}$ (we define $\boA_0:=\boD$ so $0$ is guaranteed to be in $J(T,L)$). 
Then we pick $i\in J(T,L)$ uniformly at random. Finally, we exchange $\boA_i$ with $\boD$ (if $i=0$, nothing happens) and get a new triple $(T',H,\boD'=\boA_i)$. 
We claim that the distribution of $(T',H,\boD')$ is still uniform over $\cP(M,L)$. In fact, for any $(T',H,\boD')\in \cP(M,L)$, there are exactly $|J(T',L)|$ many triples $(T,H,\boD)\in \cP(M,L)$ that could generate $(T',H,\boD')$ by following the above exchange process, and for each of those triples, the probability that $(T',H,\boD')$ is sampled is exactly $\frac{1}{|J(T',L)|}$.

Then we can deduce that
\begin{align*}
&\bE_{(T,H,\boD)\sim \cP(M,L)} [\cL-\eB(T,H,\boD)\cdot \bone_{L\in \BalL(T)}]\\
=& \bE_{(T,H,\boD)\sim \cP(M,L)} \left[\bE_{i\in J(T,L)}\left[\cL-\eB(T^{i},H,\boD^{i})\cdot \bone_{L\in \BalL(T^{i})}\right]\right],
\end{align*}
where $(T^{i},H,\boD^{i})$ is the triple induced by swapping the role of $\boD$ and $\boA_i$ in $(T,H,\boD)$.

It suffices to show for any fixed $(T,H,\boD)\in \cP(M,L)$, 
\begin{equation}\label{eqn:exchange}
\bE_{i\in J(T,L)}[\cL-\eB(T^{i},H,\boD^{i})\cdot \bone_{L\in \BalL(T^{i})}] \le \epsilon/2.    
\end{equation}

If $|J(T,L)|<\frac{m}{2c_k}$, then $L\notin \BalL(T^{i})$ for all $i\in J(T,L)$. We can then conclude that the left-hand side of \eqref{eqn:exchange} is $0$, as desired.

On the other hand, if $|J(T,L)|\ge \frac{m}{2c_k}$, we can assume that $(T^{j},H,\boD^{j})$ is bad for some $j\in J(T,L)$ because otherwise, the left-hand side of \eqref{eqn:exchange} is again $0$. Define $X_i:=\{\text{all labelings on }\boA_i$ consistent with $M_A(\boA_i)\}$, and $X:=X_0\times \cdots \times X_{m}$. Let 
\begin{equation*}
Y:=\{L_{\boA_0\cup\cdots\cup \boA_{m}}:L\in \Lall(T^{j},L_3(H\sqcup \boD^{j}))\cap \cL\}\subseteq X
\end{equation*} 
be the set of labelings $L$'s restrictions on $\boD$ and A-blocks, where $L$ is sampled w.r.t. $(T,H,\boD)$ when generating $\mu_3$ and $L\in \cL$ (recall $\cL$ is the set of columns of the fixed rectangle $\cR=\cM\times \cL$). Then we have
\begin{equation*}
|Y|\ge 2^{-\delta m}\prod_{i\ne j} |X_i|\ge 2^{-1.1\delta m} |X|,
\end{equation*} 
where the first inequality follows from the assumption that $(T^{j},H,\boD^{j})$ is not small, and the second inequality follows from the simple observation that $2^{0.1\delta m}>c_k\ge |X_j|$ when $m$ is sufficiently large.

Choosing $\delta:=0.9\delta(\frac{\epsilon}{3},\frac{\epsilon}{4c_k},c_k)$ in \Cref{lem:biased}, the number of $\frac{\epsilon}{3}$-biased indices is bounded by $\frac{\epsilon m}{4c_k}$, which implies that at most an $\frac{\epsilon}{2}$-fraction of indices $i\in J(T,L)$ are $\frac{\epsilon}{3}$-biased. 

To complete the proof, we only need to show that for any $\frac{\epsilon}{3}$-unbiased index $i$, $(T^{i},H,\boD^{i})$ is $\cL$-good. Indeed, if $i$ is $\frac{\epsilon}{3}$-unbiased, for any labeling $L'$ on $G_{\boA_i}$ consistent with $M_A(\boA_i)$, we have $p_{\cL,T^{i}}(L_3(H)\cup L')\ge \frac{1}{1+\epsilon/3}\Pr_{y\sim X}[y\in Y]$ and $p_{\cL,T^{i}}(L_3(H\cup \boD^{i}))\le (1+\epsilon/3)\Pr_{y\sim X}[y\in Y]$. Thus
\begin{equation}
    p_{\cL,T^{i}}(L_3(H\cup \boD^{i}))\le (1+\epsilon) p_{\cL,T^{i}}(L_3(H)\cup L')
\end{equation}
since $(1+\epsilon/3)^2\le 1+\epsilon$.

Therefore, for at most an $\frac{\epsilon}{2}$-fraction of indices $i\in J(T,L)$, $(T^i,H,\boD^i)$ is $\cL$-bad, which directly implies \eqref{eqn:exchange}.

In conclusion, no matter the size of $J(T,L)$, \eqref{eqn:exchange} always holds. So the lemma follows.
\end{proof}

\section{Polytope Feasibility} \label{sec:append}

In this section we prove that our relaxation is infeasible if and only if the bipartite graph $G$ has no perfect matching with an odd number of red edges (called an \textit{odd-red} PM). 
Recall that the feasibility linear program we used is the following. 

\begin{align*}
    \sum_{e\in \delta(u)} x_{e} &= 1 \quad \mbox{for all } u \in U\cup V \\
    \sum_{e \in E_L} x_{e} &\geq 1 \quad \mbox{for all } L \in \Lall(G) \\
    x_{e} &\geq 0 \quad \mbox{for all $e \in E$}.
\end{align*}

Here, labelings $\Lall(G)$ are as defined in Section~\ref{sec:relax-poly}. In other words, we will prove that $G$ has no odd-red PM if and only if there is a labeling such that all edges of $G$ that are in some perfect matching (let us call them the \textit{relevant edges}) are consistent with this labeling, proving the following.

\begin{theorem}\label{thm:feasible-poly}
The above LP is feasible if and only if $G$ has an odd-red perfect matching.
\end{theorem} 
\begin{proof}
By Claim~\ref{clm:parity}, if $G$ has an odd-red PM then there is no labeling such that all matching edges are consistent and so the indicator vector of this odd-red PM satisfies the above constraints. 
It remains to prove that if $G$ has no odd-red PM then the above polytope is empty.
We will do this by showing that if $G$ has no odd-red PM then there is a labeling consistent with all relevant edges. 
This is sufficient because any feasible point in the above LP is a convex combination of perfect matchings since the bipartite matching polytope is integral, but all the edges in these perfect matchings will be consistent with some labeling so the polytope is indeed empty. 
To this end, fix a perfect matching of $G$ that has an even number of red edges.
We show that if we cannot find a labeling consistent on all the relevant edges, then we can find an odd-red perfect matching.

Consider the edge-induced subgraph of relevant edges and take some connected component. 
We direct the edges: direct $e$ from the left partition to the right partition if $e$ is in the perfect matching and right to left otherwise.
\begin{claim}\label{clm:strong-con}
This component is strongly connected unless it is a single edge. 
\end{claim} 
\begin{proof}
For any edge in this component, it is either in every perfect matching that exists in $G$, or not.
In the former case, the component is just the single edge since no edges adjacent to either endpoint can be in any perfect matching.

In the latter case, since the edge is relevant, it must be in an alternating cycle of the perfect matching we fixed at the beginning.
In fact, then every edge in the component has this property since an edge cannot be in every perfect matching if it has an adjacent edge that is in some perfect matching.
Now that every edge is in an alternating cycle, we have that the component is strongly connected.
\end{proof}

Now it suffices to either consistently label this strongly connected component or find an odd-red alternating cycle.
Note that by rotating the matching on a single odd-red alternating cycle we change the parity of the red edges in the perfect matching.
To this end, start at some vertex in the left partition and build a tree using breadth-first search (depth-first search works the same).
Let the root node be at depth zero.
We have that even-depth vertices in the tree have a single out-going edge that is a matching edge and odd-depth vertices can have multiple out-going edges that are all non-matching since we started in the left partition. 

Fix some labeling consistent with the edges of the tree, which is possible because there are no cycles.
However, there are non-tree edges that could violate this labeling.
In particular, a non-tree edge is a non-matching edge and directs from an odd-depth vertex to an even-degree one, and is either a back-edge (directs to a lower-depth vertex on the same root-leaf path), a forward-edge (directs to a higher-depth vertex on the same root-leaf path), or a cross-edge.

Consider first the case of a back-edge $(u,v)$, as seen in Figure~\ref{fig:back-edge}.
Then we have an alternating cycle starting from $v$ that uses tree edges until $u$, and then uses the violating edge $(u,v)$, so it has exactly one violating edge.

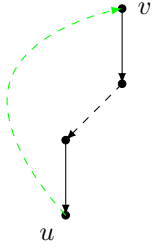
\begin{figure}
    \centering
    \begin{tikzpicture}
    \draw [fill=black] (0, 0) circle [radius = 0.05];
    \draw [fill=black] (0, -1) circle [radius = 0.05];
    \draw [fill=black] (-0.75, -1.75) circle [radius = 0.05];
    \draw [fill=black] (-0.75, -2.75) circle [radius = 0.05];
    \draw [-latex] (0, 0) -- (0,-1);
    \draw [-latex, dashed] (0, -1) -- (-0.75,-1.75);
    \draw [-latex] (-0.75, -1.75) -- (-0.75,-2.75);
    \draw [-latex, dashed, color=green] plot [smooth, tension=1] coordinates { (-0.75,-2.75) (-1.5, -1) (0,0)};
    \node at (-1, -3) {$u$};
    \node at (0.3, 0) {$v$};
    \end{tikzpicture}
    \caption{Back-edge along same root-leaf path, with the green edge violating the labeling.}
    \label{fig:back-edge}
 \end{figure} 

\begin{claim}\label{clm:vio-cycle}
An even cycle has an odd number of red edges if and only if for every labeling there are an odd number of violating edges.
\end{claim}
\begin{proof}
Note that the parity of blue edges is the same as the parity of red edges in an even cycle.
If we start with an arbitrary label at an arbitrary vertex and label the vertices consistently with the edge colors by walking around the cycle, then each blue edge flips the sign of the last labelled vertex.
Hence, an odd number of blue edges results in one violating edge and an even number of blue edges results in all edges being consistently labelled.
Now if we start flipping some labels, each flipped label can either add two violating edges if the two incident edges to the flipped vertex were originally consistent, subtract two violating edges if the incident edges were originally violating, or maintain the number of violating edges if one incident edge was violating and the other was consistent.
Hence, the parity of number of violating edges is preserved for any labeling.
\end{proof}

With this claim, we have that the alternating cycle in Figure~\ref{fig:back-edge} has an odd number of red edges so we have an odd-red matching.

\begin{figure}
    \centering
    \begin{tikzpicture}
    \draw [fill=black] (0, 0) circle [radius = 0.05];
    \draw [fill=black] (0, -1) circle [radius = 0.05];
    \draw [fill=black] (-1,-2.5) circle [radius = 0.05];
    \draw [fill=black] (-1, -1.5) circle [radius = 0.05];
    \draw [fill=black] (0,-3) circle [radius = 0.05];
    \draw [fill=black] (0, -4) circle [radius = 0.05];
    \draw [-latex] (0, 0) -- (0,-1);
    \draw [-latex, dashed] (0, -1) -- (-1, -1.5);
    \draw [-latex] (-1, -1.5) -- (-1, -2.5);
    \draw [-latex, dashed] (-1, -2.5) -- (0,-3);
    \draw [-latex] (0,-3) -- (0, -4);
    \draw [-latex, dashed, color=green] (0,-1) -- (0,-3);
    \draw [-latex, dotted, color=blue] plot [smooth, tension=1] coordinates {(0,-4) (0.2, -4.5) (1, -1.5) (0.2, 0.5) (0,0)};
    \node at (0.3, -1) {$u$};
    \node at (0.3, -3) {$v$};
    \node at (-0.9, -0.9) {lca$(u,v)$};
    \end{tikzpicture}
    \caption{Forward edge $uv$. The dotted curve in blue shows some path from $v$ to $u$.}
    \label{fig:forward-edge}
\end{figure}
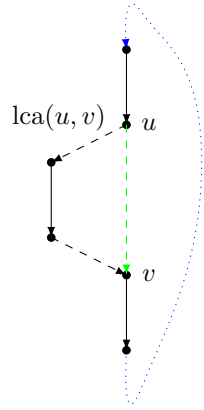

Now suppose that the violating non-tree edge $(u,v)$ is a forward edge, as in Figure~\ref{fig:forward-edge}.
Since the component is strongly connected by Claim~\ref{clm:strong-con} there must be an alternating path from $v$ to the lowest common ancestor of $u$ and $v$, lca$(u,v)$ (which in this case is just $u$).
If there are an odd number of violating edges in this path then complete a circuit by going from lca$(u,v)$ to $v$ by only taking tree-edges.
Otherwise, the completed circuit by going from lca$(u,v)$ to $v$ using the violating edge $(u,v)$ has an odd number of violating edges.
In either case, this circuit enters and exits every vertex the same number of times so the edges traversed form an Eulerian circuit.
A basic fact in graph theory says that an Eulerian circuit can be decomposed into edge-disjoint cycles. 
Since the original circuit had an odd number of violating edges it cannot be that every cycle in the decomposition has an even number of violating edges.
Taking a cycle with an odd number of violating edges, by Claim~\ref{clm:vio-cycle} it is an odd-red alternating cycle.

Actually the case of a cross non-tree edge $(u,v)$ is the same.
It can look like either tree in Figure~\ref{fig:cross-edge}, depending on the relative depths of $u$ and $v$, but we have the same argument as for forward edges.
This concludes the proof of Theorem~\ref{thm:feasible-poly}.

\begin{figure}
    \begin{minipage}[t][5.0cm][t]{0.5\textwidth}
    \begin{center}
    \begin{tikzpicture}
    \draw [fill=black] (0, 0) circle [radius = 0.05];
    \draw [fill=black] (0, -1) circle [radius = 0.05];
    \draw [fill=black] (0, -2) circle [radius = 0.05];
    \draw [fill=black] (-0.75, -1.75) circle [radius = 0.05];
    \draw [fill=black] (-0.75, -2.75) circle [radius = 0.05];
    \draw [fill=black] (0, -3) circle [radius = 0.05];
    \draw [-latex] (0, 0) -- (0,-1);
    \draw [-latex, dashed] (0, -1) -- (-0.75,-1.75);
    \draw [-latex, dashed] (0, -1) -- (0,-2);
    \draw [-latex] (0, -2) -- (0, -3);
    \draw [-latex] (-0.75, -1.75) -- (-0.75,-2.75);
    \draw [-latex, dashed, color=green] (-0.75,-2.75) -- (0, -2);
    \draw [-latex, dotted, color=blue] plot [smooth, tension=1] coordinates {(0,-3) (0.2, -3.5) (1, -1.5) (0.2, 0.5) (0,0)};
    \node at (-1, -3) {$u$};
    \node at (0.3, -2) {$v$};
    \node at (-0.9, -0.9) {lca$(u,v)$};
    \end{tikzpicture}
    \end{center}
    \end{minipage}
    \begin{minipage}[t][5.0cm][t]{0.5\textwidth}
    \begin{center}
    \begin{tikzpicture}
    \draw [fill=black] (0, 0) circle [radius = 0.05];
    \draw [fill=black] (0, -1) circle [radius = 0.05];
    \draw [fill=black] (0.75,-1.75) circle [radius = 0.05];
    \draw [fill=black] (-0.75, -1.75) circle [radius = 0.05];
    \draw [fill=black] (-0.75, -2.75) circle [radius = 0.05];
    \draw [fill=black] (0.75, -2.75) circle [radius = 0.05];
    \draw [-latex] (0, 0) -- (0,-1);
    \draw [-latex, dashed] (0, -1) -- (-0.75,-1.75);
    \draw [-latex, dashed] (0, -1) -- (0.75,-1.75);
    \draw [-latex] (0.75,-1.75) -- (0.75, -2.75);
    \draw [-latex] (-0.75, -1.75) -- (-0.75,-2.75);
    \draw [-latex, dashed] (-0.75, -2.75) -- (0,-3.5);
    \draw [-latex] (0,-3.5) -- (0,-4.5);
    \draw [-latex, dashed, color=green] (0.75,-2.75) -- (0, -3.5);
    \draw [-latex, dotted, color=blue] plot [smooth, tension=1] coordinates {(0, -4.5) (0.4, -4.75) (1.5, -2) (0.3, 0.5) (0,0)};
    \node at (1, -2.75) {$u$};
    \node at (-0.3, -3.5) {$v$};
    \node at (-0.9, -0.9) {lca$(u,v)$};
        \end{tikzpicture}
    \end{center}
    \end{minipage}
    \caption{Cross-edges. The dotted curves in blue show some path from $v$ to lca$(u,v)$.}
    \label{fig:cross-edge}
\end{figure}
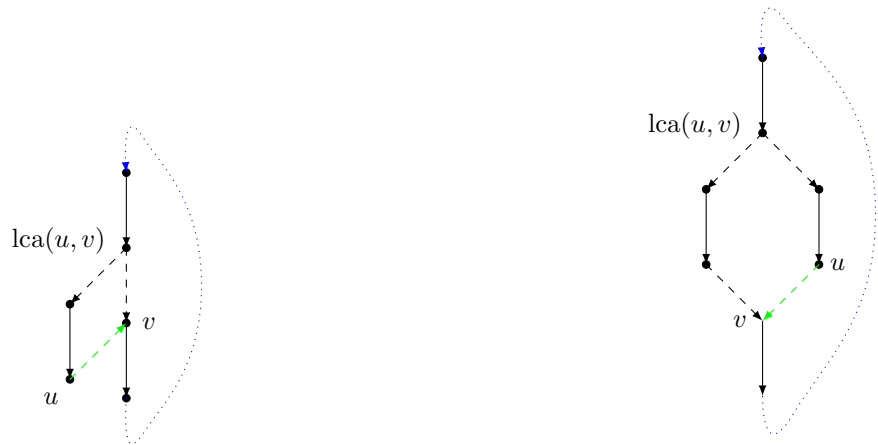
\end{proof}

\bibliography{references}

\begin{thebibliography}{10}

\bibitem{Jacobi}
Jacobi's bound.
\newblock
  \url{https://www.lix.polytechnique.fr/~ollivier/JACOBI/jacobiEngl.htm}.
\newblock Accessed: 2022-07-13.

\bibitem{AWZ17}
Stephan Artmann, Robert Weismantel, and Rico Zenklusen.
\newblock A strongly polynomial algorithm for bimodular integer linear
  programming.
\newblock In {\em Proceedings of the 49th Annual ACM SIGACT Symposium on Theory
  of Computing}, STOC 2017, page 1206–1219, New York, NY, USA, 2017.
  Association for Computing Machinery.
\newblock \href {https://doi.org/10.1145/3055399.3055473}
  {\path{doi:10.1145/3055399.3055473}}.

\bibitem{Balas1998}
Egon Balas.
\newblock Disjunctive programming: Properties of the convex hull of feasible
  points.
\newblock {\em Discrete Applied Mathematics}, 89(1):3--44, 1998.
\newblock URL:
  \url{https://www.sciencedirect.com/science/article/pii/S0166218X9800136X},
  \href {https://doi.org/https://doi.org/10.1016/S0166-218X(98)00136-X}
  {\path{doi:https://doi.org/10.1016/S0166-218X(98)00136-X}}.

\bibitem{BraunPokutta}
Gábor Braun and Sebastian Pokutta.
\newblock The matching problem has no fully polynomial size linear programming
  relaxation schemes.
\newblock {\em IEEE Transactions on Information Theory}, 61(10):5754--5764,
  2015.
\newblock \href {https://doi.org/10.1109/TIT.2015.2465864}
  {\path{doi:10.1109/TIT.2015.2465864}}.

\bibitem{Cevallos_2018}
Alfonso Cevallos, Stefan Weltge, and Rico Zenklusen.
\newblock Lifting linear extension complexity bounds to the mixed-integer
  setting.
\newblock In {\em Proceedings of the Twenty-Ninth Annual {ACM}-{SIAM} Symposium
  on Discrete Algorithms}, pages 788--807. Society for Industrial and Applied
  Mathematics, jan 2018.
\newblock URL: \url{https://doi.org/10.1137\%2F1.9781611975031.51}, \href
  {https://doi.org/10.1137/1.9781611975031.51}
  {\path{doi:10.1137/1.9781611975031.51}}.

\bibitem{DBLP:journals/corr/abs-2203-00671}
Li~Chen, Rasmus Kyng, Yang~P. Liu, Richard Peng, Maximilian~Probst Gutenberg,
  and Sushant Sachdeva.
\newblock Maximum flow and minimum-cost flow in almost-linear time.
\newblock {\em CoRR}, abs/2203.00671, 2022.

\bibitem{edm65}
Jack Edmonds.
\newblock Maximum matching and a polyhedron with $0,1$ vertices.
\newblock {\em Journal of Research of the National Bureau of Standards},
  69:125--130, 1965.

\bibitem{edm65matching}
Jack Edmonds.
\newblock Paths, trees, and flowers.
\newblock {\em Canadian Journal of Mathematics}, 17:449--467, 1965.

\bibitem{ElMaalouly}
Nicolas ElMaalouly.
\newblock Exact matching: Algorithms and related problems, 2022.
\newblock URL: \url{https://arxiv.org/abs/2203.13899}, \href
  {https://doi.org/10.48550/ARXIV.2203.13899}
  {\path{doi:10.48550/ARXIV.2203.13899}}.

\bibitem{FennerGT16}
Stephen~A. Fenner, Rohit Gurjar, and Thomas Thierauf.
\newblock Bipartite perfect matching is in {quasi-NC}.
\newblock In {\em {STOC}}, pages 754--763, 2016.

\bibitem{FMP+12}
Samuel Fiorini, Serge Massar, Sebastian Pokutta, Hans~Raj Tiwary, and Ronald
  de~Wolf.
\newblock Linear vs. semidefinite extended formulations: Exponential separation
  and strong lower bounds.
\newblock In {\em Proceedings of the Forty-Fourth Annual ACM Symposium on
  Theory of Computing}, STOC '12, page 95–106, New York, NY, USA, 2012.
  Association for Computing Machinery.
\newblock \href {https://doi.org/10.1145/2213977.2213988}
  {\path{doi:10.1145/2213977.2213988}}.

\bibitem{GKMT17}
Rohit Gurjar, Arpita Korwar, Jochen Messner, and Thomas Thierauf.
\newblock Exact perfect matching in complete graphs.
\newblock {\em ACM Trans. Comput. Theory}, 9(2), apr 2017.
\newblock \href {https://doi.org/10.1145/3041402} {\path{doi:10.1145/3041402}}.

\bibitem{GurjarT16}
Rohit Gurjar and Thomas Thierauf.
\newblock Linear matroid intersection is in quasi-{NC}.
\newblock In {\em Proceedings of the 48th Annual {ACM} {SIGACT} Symposium on
  Theory of Computing, {STOC}}, pages 821--830, 2017.

\bibitem{GurjarTV18}
Rohit Gurjar, Thomas Thierauf, and Nisheeth~K. Vishnoi.
\newblock Isolating a vertex via lattices: Polytopes with totally unimodular
  faces.
\newblock In {\em {ICALP}}, pages 74:1--74:14, 2018.

\bibitem{ImpagliazzoW97}
Russell Impagliazzo and Avi Wigderson.
\newblock \emph{P = BPP} if \emph{E} requires exponential circuits:
  Derandomizing the {XOR} lemma.
\newblock In {\em STOC}, pages 220--229, 1997.

\bibitem{ElMaalouly-Steiner}
Nicolas~El Maalouly and Raphael Steiner.
\newblock Exact matching in graphs of bounded independence number, 2022.
\newblock URL: \url{https://arxiv.org/abs/2202.11988}, \href
  {https://doi.org/10.48550/ARXIV.2202.11988}
  {\path{doi:10.48550/ARXIV.2202.11988}}.

\bibitem{MVV87}
Ketan Mulmuley, Umesh~V. Vazirani, and Vijay~V. Vazirani.
\newblock Matching is as easy as matrix inversion.
\newblock In {\em Proceedings of the Nineteenth Annual ACM Symposium on Theory
  of Computing}, STOC '87, page 345–354, New York, NY, USA, 1987. Association
  for Computing Machinery.
\newblock \href {https://doi.org/10.1145/28395.383347}
  {\path{doi:10.1145/28395.383347}}.

\bibitem{RAZBOROV1992385}
A.A. Razborov.
\newblock On the distributional complexity of disjointness.
\newblock {\em Theoretical Computer Science}, 106(2):385--390, 1992.
\newblock URL:
  \url{https://www.sciencedirect.com/science/article/pii/030439759290260M},
  \href {https://doi.org/https://doi.org/10.1016/0304-3975(92)90260-M}
  {\path{doi:https://doi.org/10.1016/0304-3975(92)90260-M}}.

\bibitem{Rothvoss17}
Thomas Rothvoss.
\newblock The matching polytope has exponential extension complexity.
\newblock {\em J. ACM}, 64(6), sep 2017.
\newblock \href {https://doi.org/10.1145/3127497} {\path{doi:10.1145/3127497}}.

\bibitem{Sinha}
Makrand Sinha.
\newblock Lower bounds for approximating the matching polytope.
\newblock In {\em Proceedings of the Twenty-Ninth Annual ACM-SIAM Symposium on
  Discrete Algorithms}, SODA '18, page 1585–1604, USA, 2018. Society for
  Industrial and Applied Mathematics.

\bibitem{nST17}
Ola Svensson and Jakub Tarnawski.
\newblock The matching problem in general graphs is in quasi-nc.
\newblock In {\em {FOCS}}, pages 696--707, 2017.

\bibitem{Yannak88}
Mihalis Yannakakis.
\newblock Expressing combinatorial optimization problems by linear programs.
\newblock In {\em Proceedings of the Twentieth Annual ACM Symposium on Theory
  of Computing}, STOC '88, page 223–228, New York, NY, USA, 1988. Association
  for Computing Machinery.
\newblock \href {https://doi.org/10.1145/62212.62232}
  {\path{doi:10.1145/62212.62232}}.

\end{thebibliography}


\end{document}